\documentclass[12pt,a4paper,notitlepage]{article}

 \usepackage[cmex10]{amsmath} % Use the [cmex10] option to ensure complicance
\usepackage{amssymb,mathrsfs,dsfont}
\usepackage{color}
\usepackage{graphicx}
\usepackage{amsthm}
\usepackage{soul}

\newtheorem{theorem}{Theorem}
\newtheorem{lemma}{Lemma}
\newtheorem{proposition}{Proposition}
\newtheorem{corollary}{Corollary}
\theoremstyle{definition}
\newtheorem{remark}{Remark}

\DeclareMathOperator{\Var}{\mathsf{Var}}
\DeclareMathOperator{\E}{\mathbb{E}}

\newcommand{\ofer}[1]{\textcolor{red}{[Ofer: #1]}}  
\newcommand{\uri}[1]{\textcolor{blue}{[Uri: #1]}}  
\definecolor{mygreen}{RGB}{0, 150, 50}

\newcommand{\liu}[1]{\textcolor{blue}{[Liu: #1]}}

\definecolor{myred}{RGB}{150, 50, 0}

\newcommand{\vect}[1]{\mathbf{#1}}

\DeclareMathOperator*{\argmax}{argmax}

\def\fisher{\mathrm{I}_F}

\def\tn{{\otimes n}}

\def\PP{\mathbb{P}}
\def\EE{\mathbb{E}}
\def\Ber{\mathrm{Ber}}

\title{Communication Complexity of Estimating Correlations}
\author{U. Hadar, J. Liu, Y. Polyanskiy, O. Shayevitz \thanks{Order of authors is alphabetical. U.H. and O.S. \{emails: urihadar@mail.tau.ac.il, ofersha@eng.tau.ac.il\} are with the Department of Electrical Engineering--Systems, Tel Aviv University, Tel Aviv, Israel. J.L. and Y.P. \{emails: jingbo@mit.edu, yp@mit.edu\} are with the Institute for Data, Systems, and Society and the Department of Electrical Engineering and Computer Science, Massachusetts Institute of Technology, Cambridge, MA 02139, USA.}}

\date{}

\begin{document}
\maketitle

\begin{abstract}
We characterize the communication complexity of the following distributed estimation problem. Alice and Bob observe infinitely many iid copies of $\rho$-correlated unit-variance (Gaussian or $\pm1$ binary) random variables, with unknown $\rho\in[-1,1]$. By interactively exchanging $k$ bits, Bob wants to produce an estimate $\hat\rho$ of $\rho$. We show that the best possible performance (optimized over interaction protocol $\Pi$ and estimator $\hat \rho$) satisfies $\inf_{\Pi \hat\rho}\sup_\rho \mathbb{E} [|\rho-\hat\rho|^2] = \tfrac{1}{k} (\frac{1}{2 \ln 2} + o(1))$. Curiously, the number of samples in our achievability scheme is exponential in $k$;
by contrast, a naive scheme exchanging $k$ samples achieves the same $\Omega(1/k)$ rate but with a suboptimal prefactor.
Our protocol achieving optimal performance is one-way (non-interactive).
We also prove the $\Omega(1/k)$ bound even when $\rho$ is restricted to any small open sub-interval of $[-1,1]$ (i.e. a local minimax lower bound). 
Our proof techniques rely on symmetric strong data-processing inequalities and various tensorization techniques from information-theoretic interactive common-randomness extraction.
Our results also imply an $\Omega(n)$ lower bound on the information complexity of the Gap-Hamming problem, for which we show a direct information-theoretic proof. 
\end{abstract}

\maketitle
\clearpage

\section{Introduction}
The problem of distributed statistical inference under communication constraints has gained much recent interest in the  theoretical computer science, statistics, machine learning, and information theory communities. The prototypical setup involves two or more remote parties, each observing local samples drawn from of a partially known joint statistical model. The parties are interested in estimating some well-defined statistical property of the model from their data, and to that end, can exchange messages under some prescribed communication model. The communication complexity associated with this estimation problem concerns the minimal number of bits that need to be exchanged in order to achieve a certain level of estimation accuracy. Whereas the sample-complexity of various estimation problems in the centralized case is well studied (see e.g. \cite{lehmann2006theory},\cite{van2004detection}), the fundamental limits of estimation in a distributed setup are far less understood, due to the inherent difficulty imposed by the restrictions on the communication protocol. 

In this paper, we study the following distributed estimation problem. Alice and Bob observe infinitely many iid copies of $\rho$-correlated unit variance random variables, that are either binary symmetric or Gaussian, and where the correlation $\rho\in[-1,1]$ is unknown. By interactively exchanging $k$ bits on a shared blackboard, Bob wants to produce an estimate $\hat\rho$ that is guaranteed to be $\epsilon$-close to $\rho$ (in the sense that $\EE[(\hat\rho - \rho)^2] \le \epsilon^2$) regardless of the true underlying value of the correlation. We show that the communication complexity of this task, i.e., the minimal number of bits $k$ that need to be exchanged between Alice and Bob to that end, 
is $\frac{1+o(1)}{2\epsilon^2\ln2}$ in both the binary and Gaussian settings, and one-way schemes are optimal. 
We also prove a local version of the bound, showing that the communication complexity is still $\Theta_\rho(1/\epsilon^2)$ even if the real correlation is within an interval of vanishing size near $\rho$.

Let us put our work in context of other results in the literature. The classical problem of communication complexity, originally introduced in a seminal paper by Yao for two parties \cite{Yao:1979:CQR:800135.804414}, has been extensively studied in various forms and variations, see e.g. \cite{kushilevitz_nisan_1996} and references therein. In its simplest (two-party) form, Alice and Bob wish to compute some given function of their local inputs, either exactly for any input or with high probability over some distribution on the inputs, while interactively exchanging the least possible number of bits. While in this paper we also care about the communication complexity of the  task at hand, our setting differs from the classical setup in important ways. First, rather than computing a specific function of {\em finite} input sequences with a small error probability under a known distribution (or in the worst case), we assume an unknown parametric distribution on {\em infinite} input sequences, and wish to approximate the underlying parameter to within a given precision. In a sense, rather than to compute a function, our task is to interactively extract the most valuable bits from the infinite inputs towards our goal. Notwithstanding the above, an appealing way to cast our problem is to require interactively approximating the function
\begin{align}
    f(\vect{X},\vect{Y}) = \lim_{n\to \infty} \frac{1}{n} \sum_{i=1}^n X_i Y_i \,
\end{align}
for infinite iid (binary or Gaussian) strings $\vect{X},\vect{Y}$, to within precision $\epsilon$, which we show requires $\Theta(1/\epsilon^2)$ bits of interactive communication.

In another related aspect, many traditional communication complexity lower bounds are proved via information-theoretic arguments, most notably by bounding the {\em information complexity} of good protocols over a suitable choice of distribution over the inputs, see e.g. the classical proof for the disjointness problem \cite{bar2004information}. Our proofs have a similar information-theoretic flavor; in fact, our key technical contribution is connecting a so-called {\em symmetric strong data-processing inequality (SDPI)},
previously considered in \cite{liu2016common} in the context of interactive secret key generation,
to interactive hypothesis testing and estimation problems. 
Loosely speaking, the symmetric SDPI gives the following bound:
\begin{align*}
\begin{array}{cc}
\textbf{mutual information {\em interchanged} between Alice and Bob}& \\
\le& \\
\rho^2 \times\textbf{mutual information {\em injected} by Alice and Bob}&
\end{array}
\end{align*}
This is formalized in our Lemmas~\ref{lem3}, \ref{lem4} and \ref{lem11}, 
where the upper and lower expressions above correspond to $R$ in eq.~\eqref{eq:R} and $S$ in eq.~\eqref{eq:S}, respectively. In fact, as a side application of this interactive SDPI, we show an $\Omega(n)$ lower bound on information complexity of the Gap-Hamming problem~\cite{indyk2003tight},\cite{chakrabarti2012optimal}, which has so far resisted an information-theoretic attack; see Remark~\ref{rem:GH} for details. 

There has also been much contemporary interest in distributed estimation with communication constraints under a different context, where a finite number of iid samples from a distribution belonging to some parametric family are observed by multiple remotely located parties, which in turn can communicate with a data center (either one-way or interactively) in order to obtain an estimate of the underlying parameters, under a communication budget constraint, see e.g. \cite{zhang2013information}, \cite{braverman2016communication}, \cite{han2018geometric}. These works are markedly different from ours: the samples observed by the parties are taken from the {\em same} distribution, and the main regime of interest is typically where the dimension of the problem is relatively high compared to the number of local samples (so that each party alone is not too useful), but is low relative to the total number of samples observed by all parties (so that centralized performance is good). The goal is then to communicate efficiently in order to approach the centralized performance. This stands in contrast to our case, where each party observes an {\em unlimited} number of samples drawn from a {\em different} component of a bivariate, single-parameter distribution, and the difficulty hence lies in the fact that the quantity of interest (correlation) is truly distributed; none of the parties can estimate it alone, and both parties together can estimate it to arbitrary precision in a centralized setup. Hence, the only bottleneck is imposed by communications. 

Another line of works closely related to ours has appeared in the information theory literature, limited almost exclusively to one-way protocols. The problem of distributed parameter estimation under communication constraints has been originally introduced in~\cite{zhang1988estimation}, where the authors provided a rate-distortion-type upper bound on the quadratic error in distributively estimating a scalar parameter using one-way communication (possibly to a third party) under a rate constraint in communication-bits per sample, for a limited set of estimation problems. They have studied our Gaussian setup in particular, and the upper bounds we briefly discuss herein can be deduced (albeit non-constructively) from their work (in~\cite{hadar2018distributed} it is shown how to constructively attain the same performance, and also generalize to the vector parameter case). There has been much followup work on this type of problems, especially in the discrete samples case, see~\cite{amari1998statistical} and references therein. A local and global minimax setup similar to our own (but again for one-way protocols) has been explored in~\cite{ahlswede1990minimax}. The local minimax bound we obtain (for one-way protocols) was essentially claimed in that paper, but a subtle yet crucial flaw in their proof of the Fisher information tensorization has been pointed out later in~\cite{amari1998statistical}. 

Finally, it is worth noting the closely related problem of distributed hypothesis testing for independence under communication constraints. In~\cite{ahlswede1986hypothesis}, the authors provide an exact asymptotic characterization of the optimal tradeoff between the rate (in bits per sample) of one-way protocols and the false-alarm error exponent attained under a vanishing mis-detect probability. This result has recently been extended to the interactive setup with a finite number of rounds~\cite{xiang2013interactive}.

In fact, some of our lower bounds are also based on a reduction to testing independence with finite communication complexity. For a special case of one-way protocols, this problem was recently analyzed in~\cite{sahasranand2018extra}. There is also an inherent connection between the problem of testing independence and generating common randomness from correlated sources, cf.~\cite{tyagi2015converses}, as well as between the problem of testing independence and hypercontractivity~\cite{polyanskiy2012hypothesis}. For common randomness, two recent (and independent) works~\cite{guruswami2016tight,liu2016common} showed that to (almost) agree on $L$ (almost) pure bits the minimal two-way communication required is $(1-\rho^2)L$. There are several differences between the results and techniques in these two works. The work~\cite{guruswami2016tight} followed upon earlier~\cite{canonne2017communication} and considers exponentially small probability of error. Their main tool is hypercontractivity (Concurrently, hypercontractivity bounds for one-way protocols in similar common randomness generation models were also obtained independently in \cite{lcv_isit_2015}\cite{lccv_isit_2016}). The lower bound in~\cite{guruswami2016tight} was partial, in the sense that the common randomness generated by Alice was required to be a function of her input and not of the transcript.  
Thus
\cite{guruswami2016tight}
\cite{lcv_isit_2015}\cite{lccv_isit_2016} all concern settings where one-way protocols are optimal.
In contrast, the work~\cite{liu2016common} followed on a classical work on interactive compression~\cite{kaspi1985two} and showed an unrestricted lower bound.
In that setting, one-way communication is not optimal for general sources (although it was numerically verified and proved in the limiting cases that one-way protocols are optimal for binary symmetric sources). %\liu{added some more comments}
The main tool in \cite{liu2016common} in the small communication regime was the ``symmetric strong data-processing inequality''. Here we adapt the latter to our problem.

\textit{Organization.} In Section~\ref{sec:main_results} we formally present the problem and state our main results. Section~\ref{sec:prelim} contains necessary mathematical background. 
Section~\ref{sec_b2g} proves that the achievability in the Gaussian case implies the achievability in the binary symmetric case (so that we only need to the achievability for Gaussian and converse for binary).
Section~\ref{sec_ub} proves the upper bounds.
Section~\ref{sec_lb1} proves the lower bounds in the special case of one-way protocols (as a warm-up),
and 
Section~\ref{sec_interactive} proves the lower bound in the full interactive case,
both for the global risks.
Section~\ref{sec_cr} discusses how to extend to the local version by using common randomness.
Section~\ref{sec_proof_ssdpi} gives the technical proof for the symmetric strong data processing inequality in the binary and Gaussian cases.

\section{Main results}\label{sec:main_results}

We define the problem formally as follows. Alice and Bob observe $\vect{X}$ and $\vect{Y}$ respectively, where $(\vect{X},\vect{Y})\sim P_{X Y}^\tn$.
% and $n$ is arbitrarily large (possibly infinite). 
The distribution $P_{X Y}$ belongs to one of the two families, parameterized by a single parameter $\rho\in[-1,1]$:
\begin{enumerate}
    \item Binary case: $X,Y\in\{\pm1\}$ are unbiased and $\PP[X=Y] = \tfrac{1+\rho}{2}$.
    \item Gaussian case: $X,Y$ are unit-variance $\rho$-correlated Gaussian.
\end{enumerate}
The communication between Alice and Bob proceeds in rounds: First, Alice writes $W_1 = f_1(\vect{X})$ on the board. Bob then writes $W_2 = f_2(\vect{Y},W_1)$ and so on where in the $r$-th round Alice writes $W_r$ if $r$ is odd, and Bob writes $W_r$ if $r$ is even, where in both cases $W_r = f_r(\vect{X},W_1,\ldots,W_{r-1})$.
%\begin{itemize}
%    \item First, Alice writes $W_1 = f_1(\vect{X})$ on the board.
%    \item Second, Bob writes $W_2 = f_2(\vect{Y},W_1)$ on the board.
%    \item In the $r$-th ($r$ odd) round Alice writes $W_r = f_r(\vect{X},W_1,\ldots,W_{r-1})$.
%    \item In the $r$-th ($r$ even) round Bob writes $W_r = f_r(\vect{Y},W_1,\ldots,W_{r-1})$.
%\end{itemize}
We note that, in principle, we allow each function $f_r$ to also depend on a private randomness (i.e. $f_r$ can be a stochastic map of its arguments).
We also note that our impossibility results apply to a slightly more general model where there is also a common randomness in the form of a uniform $W_0$ on $[0,1]$ pre-written on the board, but we do not need this for our algorithms.

Let $\Pi=(W_1,W_2,\ldots)$ be the contents of the board after all of (possibly infinitely many) rounds. We say that the protocol is {\em $k$-bit} if the entropy $H(\Pi) \le k$ for any $\rho \in[-1,1]$. Note that the protocol is completely characterized by the conditional distribution $P_{\Pi|\vect{X} \vect{Y}}$.

At the end of communication, Bob produces an estimate $\hat{\rho}(\Pi,\vect{Y})$ 
for the correlation $\rho$ of the underlying distribution. We are interested in characterizing the tradeoff between the communication size $k$ and the worst-case (over $\rho$) squared-error, which we call {\em quadratic risk}, in the regime where the number of samples $n$ is arbitrarily large but $k$ is fixed. Explicitly, the quadratic risk of the protocol $\Pi$ and the estimator $\hat\rho$ is given by 
\begin{align}\label{eq:riskDef1}
	R_\rho(\Pi,\hat{\rho}) \triangleq \E_\rho \left(\hat{\rho}(\Pi,\vect{Y}) - \rho\right)^2,
\end{align}
where $\E_\rho$ is the expectation under the correlation value $\rho$. Similarly, we write $P^{\rho}_{\vect{X} \vect{Y} \Pi}$ 
for the joint distribution corresponding to a fixed value of $\rho$. The (global) {\em minimax risk} is defined as
\begin{align}
R^*\triangleq \inf_{n,\Pi,\hat{\rho}}\,\sup_{-1\le \rho\le 1} R_\rho(\Pi,\hat{\rho}), \end{align}
whereas the {\em local minimax risk} is 
\begin{align}
R^*_{\rho, \delta} \triangleq \inf_{n,\Pi,\hat{\rho}}\, \sup_{|\rho'-\rho|\leq \delta} R_{\rho'}(\Pi,\hat{\rho}).
\label{e_local}
\end{align}
The infima in both the definitions above are taken over all $k$-bit protocols $\Pi$ and estimators $\hat\rho$, as well as the number of samples $n$. We will also discuss \textit{one-way protocols}, i.e. where $\Pi=W_1$ consists of a single message from Alice to Bob. We denote the global and local minimax risk in the one-way case by $R^{*1}$ and $R^{*1}_{\rho,\delta}$ respectively. 
% We will further also care about \textit{unbiased estimators}, i.e., these for which $\EE_\rho \hat\rho = \rho$ for all $\rho\in[-1,1]$, and write $R^{*,\textnormal{unbiased}}$ etc. for the relevant quantities.   

Our main results are the following.
\begin{theorem}[Upper bounds]\label{thm_ub} In \emph{both} the Gaussian \emph{and} the binary symmetric cases with infinitely many samples,
	\begin{align}\label{eq:interactive_local_risk_UB}
		R^*_{\rho,\delta} \le \frac{1}{k}\left(\frac{\left(1 - \rho^2\,\right)^2}{2 \ln 2} + o(1)\right),
	\end{align}
	as long as $\delta=o(1)$ (here and after, $o(1)$ means a vanishing sequence indexed by $k$),
	and 
	\begin{align}\label{eq:glblMinMx}
        R^* \le \frac{1}{k} \left(\frac{1}{2 \ln 2} + o(1) \right).
	\end{align}
	In fact, one-way protocols achieve these upper bounds.
\end{theorem}
\begin{remark}
Previously, \cite{hadar2018distributed} showed that there exists a one-way protocol and an unbiased estimator achieving $R_{\rho}(\Pi,\hat{\rho}) \le \frac{1}{k}\left(\frac{1 - \rho^2}{2 \ln 2} + o(1)\right)$ for any $\rho$.
The protocol (in the Gaussian case) sends the index $\argmax_{1\le i\le 2^k}X_i$ using $k$ bits and employs the super concentration property of the max.
Here, the local risk bound \eqref{eq:interactive_local_risk_UB} is tighter because we can send the index more efficiently using the side information $Y^{2^k}$ and the knowledge of $\rho$ within $o(1)$ error. Such a scheme has the drawback that it is specially designed for a small interval of $\rho$ (as in the definition of the local risk),
and hence the performance may be poor outside that small interval.
However, we remark that one can achieve the risk $\frac{1}{k}\left(\frac{\left(1 - \rho^2\,\right)^2}{2 \ln 2} + o(1)\right)$ at any $\rho$ by a \emph{two-way} protocol.
Indeed, Alice can use the first round to send $\omega(1)\cap o(k)$ bits to Bob so that Bob can estimate $\rho$ up to $o(1)$ error.
Then Bob can employ the one-way local protocol in \eqref{eq:interactive_local_risk_UB} for the $\rho$ estimated from the first round.
\end{remark}

\begin{theorem}[lower bounds]\label{thm_lb}In both the Gaussian and binary symmetric cases with infinitely many samples,
	\begin{align}\label{eq:lclMinMx}
	    R^{*}_{\rho, \delta} \ge \frac{(1 - |\rho|)^2}{2 k \ln 2} (1+o(1)).
	\end{align}
	In particular, since the global risk dominates the local risk at any $\rho$, we have
		\begin{align}\label{eq:glblMinMx1}
        R^* \ge \frac{1}{k} \left(\frac{1}{2 \ln 2} + o(1) \right).
	\end{align}
\end{theorem}

% This theorem is first proved in the special case of one-way protocols in Section~\ref{sec:one-way}. The full proof for interactive case is given in Section~\ref{sec:glob_proof}.

Note in particular that theorems Theorem~\ref{thm_ub} and \ref{thm_lb} have identified the exact prefactor in the global risk.

\begin{remark}[Unbiased estimation]
We note that the proof of Theorem~\ref{thm_lb} also implies that for any \emph{unbiased} estimator $\hat{\rho}$ of $\rho$ in the binary case it holds that
\begin{align}\label{eq:crlb_new}
    \Var \hat{\rho} \ge \frac{(1 - |\rho|)^2}{2 k \ln 2}.
\end{align}
We further note that an unbiased estimator with $\Var \hat{\rho} = (1 - \rho^2 + o(1))/(2 k \ln 2)$ was introduced in \cite{hadar2018distributed} (and discussed in Section \ref{sec_ub} below), establishing the tightness of \eqref{eq:crlb_new} at $\rho = 0$.

The bound \eqref{eq:crlb_new} follows from the Cram\'er-Rao inequality (see e.g.~\cite{van2004detection}) along with the bound we obtain for the Fisher information given in \eqref{eq:fishBndIntrctve}. The associated regularity conditions are discussed in Remark \ref{rmrk_regularity}.
\end{remark}

\begin{remark}[Gap-Hamming Problem] \label{rem:GH}
In the Gap-Hamming problem~\cite{indyk2003tight}, Alice and Bob are given binary length-$n$ vectors $(\vect{X}$ and $\vect{Y})$ respectively, with the promise that $\#\{i: X_i \neq Y_i\}$ is either $\le n/2 - \sqrt{n}$ or $\ge n/2 + \sqrt{n}$. They communicate in (possibly infinitely many) rounds to distinguish these two hypotheses. It was shown in~\cite{chakrabarti2012optimal} (later with simplified proofs in~\cite{vidick2012concentration}, \cite{sherstov2012communication}) that the communication complexity of any protocol that solves Gap-Hamming with small error probability is $\Omega(n)$ (an upper bound of $O(n)$ is trivial). However, whereas many interesting functions in communication complexity have information-theoretic lower bounds, Gap-Hamming has so far resisted an information-theoretic proof, with the exception of the single-round case for which a proof based on SDPI is known~\cite{garg_private_comm}. It is nevertheless already known that the information complexity of Gap-Hamming is linear, i.e., that $I( \Pi ; \vect{X},\vect{Y} )  = \Omega(n)$ for any $\Pi$ that solves it, under the uniform distribution on $(\vect{X},\vect{Y})$. This is however observed only indirectly, since the smooth rectangle bound used in the original proof is known to be ``below'' the information complexity, i.e., any lower bound proved using the smooth rectangle bound also yields a lower bound on information complexity. It is therefore of interest to note that our result in particular directly implies a $\Omega(n)$ lower bound on the information complexity of Gap-Hamming (and hence also on its communication complexity). 

To see this, we note that the main step in proving our main result is the inequality 
    \begin{equation}
        D(P^{\rho}_{\vect{X} \Pi} \| P^{0}_{\vect{X} \Pi}) \le \rho^2 I(\Pi; \vect{X},\vect{Y})\,. \label{eq:mainkl}
    \end{equation}
    which is implied by Theorems~\ref{thm5} and~\ref{thm6}, in Section~\ref{sec_interactive}. We note that it implies the $\Omega(n)$ lower-bound on distributional communication \textit{and information} complexity of the Gap-Hamming problem, see~\cite{chakrabarti2012optimal} for references and the original proof of $\Omega(n)$. Indeed, let $U \sim \Ber(1/2)$ and given $U$ let $\vect{X},\vect{Y}$ have correlation $\rho=(-1)^U \rho_0$, where $\rho_0=\tfrac{100}{\sqrt{n}}$. Take $\Pi$ to be a protocol used for solving the Gap-Hamming problem (which decides whether $\#\{i: X_i \neq Y_i\}$ is $\le n/2 - \sqrt{n}$ or $\ge n/2 + \sqrt{n}$ with small error probability). Its decision should equal $U$ with high probability, and hence there exists a decision rule based on $\Pi$ reconstructing $U$ with high probability of success. Thus $I(U;\Pi) =\Omega(1)$, and we further have
    \begin{align}
        I(U;\Pi) \le I(U; \Pi, \vect{X}) \le \frac{1}{2} D(P^{+\rho_0}_{\vect{X} \Pi} \| P^{0}_{\vect{X} \Pi}) + \frac{1}{2} D(P^{-\rho_0}_{\vect{X} \Pi} \| P^{0}_{\vect{X} \Pi})\,, 
    \end{align}
    where the last inequality follows from a property of the mutual information (\eqref{eq:mirad} below). Finally, from~\eqref{eq:mainkl} we get the statement that $H(\Pi) \ge \Omega(\rho_0^{-2}) = \Omega(n)$. 
    
    As pointed out by the anonymous reviewer, a more general version of the Gap-Hamming problem concerns the decision between $\#\{i\colon X_i\neq Y_i\}\le n/2-g$ and $n/2+g$ for some $\sqrt{n}\le g\le n/2$, and it was shown in~\cite[Proposition~4.4]{chakrabarti2012optimal} that the communication complexity is $\Omega(n^2/g^2)$. This result can also be recovered by the above argument. And notably, this result also implies the $R^* = \Omega(1/k)$ lower bound.
\end{remark}

\section{Preliminaries}\label{sec:prelim}
\subsection{Notation}
Lower- and upper-case letters indicate deterministic and random variables respectively, with boldface used to indicate $n$-dimensional vectors. 
For any positive integer $r$, the set $\{1,2,\dots,r\}$ is denoted by $[r]$. 
Let $P$ and $Q$ be two probability distributions over the same probability space. The {\em KL divergence} between $P$ and $Q$ is
\begin{align}
    D(P\|Q) = \int \log\left(\frac{dP}{dQ}\right)dP
\end{align}
with the convention that $D(P\|Q) = \infty$ if $P$ is not absolutely continuous w.r.t. $Q$. Logarithms are taken to the base $2$ throughout, unless otherwise stated. With this definition, the mutual information between two jointly distributed r.v.s $(X,Y)\sim P_{X Y}$ can be defined
\begin{equation}
I(X;Y) = D(P_{X Y} \| P_X \times P_Y), \label{eq:midkl}
\end{equation} 
and it satisfies the ``radius'' property:
\begin{equation}
    I(X;Y) = \inf_{Q_Y} D(P_{Y|X}\|Q_Y|P_X),  \label{eq:mirad}
\end{equation}
where the conditioning means taking the expectation of the conditional KL divergence w.r.t. $P_X$. Given a triplet of jointly distributed r.v.s $(X,Y,Z)$, the conditional mutual information between $X$ and $Y$ given $Z$ is 
\begin{equation}
I(X;Y | Z) = D(P_{X Y|Z} \| P_{X|Z} \times P_{Y|Z} | P_Z) 
\end{equation} 
We will say that r.v.s $A,B,C$ form a Markov chain $A-B-C$, if $A$ is independent of $C$ given $B$.

\subsection{Symmetric strong data-processing inequalities}
Given $P_{X Y}$, the standard data processing inequality states that $I(U;Y)\le I(U;X)$ for any $U$ satisfying $U-X-Y$.
Recall that a \emph{strong data processing inequality} (see e.g.\ \cite{polyanskiy2017strong}) is satisfied if there exists $s\in[0,1)$ depending on $P_{XY}$ such that $I(U;Y)\le sI(U;X)$ for any $U$ satisfying $U-X-Y$.

The connection between the strong data processing and communication complexity problems is natural, and $U$ can be thought of as the message from Alice to Bob,
$I(U;X)$ the communication complexity, and $I(U;Y)$ the information for the estimator.
However, the best constant $s$ in the strong data processing inequality is not symmetric (i.e.\ $s(P_{XY})= s(P_{YX})$ is not true for general $P_{XY}$),
whereas the performance in an interactive communication problems is by definition symmetric w.r.t.\ the two parties.
An inequality of the following form, termed ``symmetric strong data processing inequality'' in \cite{liu2016common},
plays a central role in interactive communication problems:
\begin{align}
&\quad I(U_1;Y)+I(U_2;X|U_1)+I(U_3;Y|U^2)+\dots
\nonumber
\\
&\le s_{\infty}[I(U_1;X)+I(U_2;Y|U_1)+I(U_3;X|U^2)+\dots]
\label{e_ssdpi}
\end{align}
where $U_1$, $U_2$, \dots must satisfy
\begin{align}
U_r-(X, U^{r-1})-Y,\quad &r\in\{1,2,\dots\}\setminus 2\mathbb{Z},\label{e_markov1}
\\
U_r-(Y, U^{r-1})-X,\quad &r\in\{1,2,\dots\}\cap 2\mathbb{Z},\label{e_markov2}
\end{align}
and where $s_{\infty}$ depends only on $P_{XY}$.
%Clearly $s_{\infty}\ge s$, and $s_{\infty}(P_{XY})=s_{\infty}(P_{YX})$.
Clearly $s_{\infty}(P_{XY})=s_{\infty}(P_{YX})$ and $s_{\infty}\ge s$.
A succinct characterization of $s_{\infty}$ in terms of the ``marginally convex envelope'' was reported in \cite{liu2016common}.
Using the Markov assumptions \eqref{e_markov1}-\eqref{e_markov2} we can also rewrite \eqref{e_ssdpi} as 
\begin{align}
I(X;Y)-I(X;Y|U)\le s_{\infty}I(U;X,Y).
\end{align}

When $X,Y$ are iid binary symmetric vectors with correlation $\rho^2$ per coordinate,
it was shown in \cite{liu2016common} that $s_{\infty}=\rho^2$, equal to the strong data processing constant.
In this paper, we extend the result to Gaussian vectors with correlation $\rho$ per coordinate (Theorem~\ref{thm6}).

In order to upper bound $s_{\infty}$ in the binary case, \cite{liu2016common} observed that $s_{\infty}(Q_{XY})$ is upper bounded by the supremum of $s(P_{XY})$ over $P_{XY}$ a ``marginally titled'' version of $Q_{XY}$.
Indeed, note that the Markov structure in the strong data processing inequality implies that $P_{XY|U=u}(x,y)=f(x)P_{XY}(x,y)$ for some function $f$.
In the case of symmetric strong data processing inequality, 
the Markov conditions \eqref{e_markov1} and \eqref{e_markov2} imply that
\begin{align}
P_{XY|U^r=u^r}(x,y)=f(x)P_{XY}(x,y)g(y),
\end{align}
which naturally lead one to considering the following result:
\begin{lemma}[{\cite[Theorem~6]{liu2016common}}]\label{lem_symmetric}
Let $Q_{XY}$ be the distribution of a binary symmetric random variables with correlation $\rho \in [-1,1]$, i.e. $Q_{XY}(x,y) = \tfrac{1}{4}(1+ (-1)^{1\{x\neq y\}}\rho)$ for $x,y \in \{0,1\}$. Let $(X,Y) \sim P_{X Y}$ have an arbitrary distribution of the form
$$ P_{X Y}(x,y) = f(x) g(y) Q_{X Y}(x,y)\,.$$
Then for any $U-X-Y-V$ we have
\begin{align}
I(U;Y) &\le \rho^2 I(U;X)\\
I(X;V) & \le \rho^2 I(Y;V)\,.
\label{eq:R1}
\end{align}
\end{lemma}
Lemma~\ref{lem_symmetric} was proved in \cite{liu2016common} by exploring the connection to the maximal correlation coefficient. 
In Section~\ref{sec_binary} we give another proof using properties of the strong data processing inequalities~\cite{polyanskiy2017strong}.

\iffalse
\begin{remark}
The proof of Lemma~\ref{lem_symmetric}  in \cite[Theorem~6]{liu2016common} proceeds by first observing that for any (not necessarily binary) $Q_{XY}$, the right side of  \eqref{eq:R1} equals the supremum of the square of maximal correlation of $P_{XY}$ over $P_{XY}$ of the form \eqref{eq:S}. 
Then direct computations can be done for the binary symmetric $Q_{XY}$.
If we only want to prove Lemma~\ref{lem3} for $r=2$, then it suffices to use the post-SDPI
(see e.g.\ \cite[Remark~8]{polyanskiy2017strong}) instead of 
 Lemma~\ref{lem_symmetric}.
\end{remark}
\fi

\subsection{Fisher information and Cram\`er-Rao inequalities}\label{subsec:fisher}
We recall some standard results from parameter estimation theory. Let $\theta$ be a real-valued parameter taking an unknown value in some interval $[a,b]$. We observe some random variable (or vector) $X$ with distribution $P(x|\theta)$ parameterized by $\theta$.
% Assume that $\{P(x|\theta)\}_{\theta\in[a,b]}$ are all absolutely continuous w.r.t. some fixed reference measure $\mu$. We say that $P(x|\theta)$ is  {\em regular} if 
% \begin{itemize}
%     \item $\frac{dP(\cdot|\theta)}{d\mu}(x)$ is twice continuously differentiable w.r.t. $\theta\in(a,b)$ for $\mu$-a.s. $x$. 
%     \item It holds that 
%     \begin{align}
%         \frac{d^m}{d\theta^m}\int e(x) dP(x|\theta) = \int e(x)\left(\frac{\partial^m}{\partial\theta^m}\frac{dP(\cdot|\theta)}{d\mu}(x)\right)d\mu(x) 
%     \end{align}
%     for any measurable $e(x)$ and $m=1,2$, such that the r.h.s. above exists. 
% \end{itemize}

Assume that $P(\cdot|\theta)$ is absolutely continuous with respect to a reference measure $\mu$, for each $\theta\in [a,b]$,
and $\frac{d P(\cdot|\theta)}{d\mu}(x)$ is differentiable with respect to $\theta\in(a,b)$ for $\mu$-almost all $x$.
Then the {\em Fisher information} of $\theta$ w.r.t.~$X$, denoted as $\fisher(X; \theta)$, is the variance
of the derivative of the log-likelihood w.r.t. $\theta$, 
% which for a regular $P(x|\theta)$ is more conveniently given by 
\begin{align}
    \fisher(X; \theta) \triangleq
    \int\left(\frac{\partial}{\partial\theta}\ln \frac{dP(\cdot|\theta)}{d\mu}(x)\right)^2 dP(x|\theta).
    % -\int\frac{\partial^2}{\partial\theta^2}\log \left(\frac{dP(\cdot|\theta)}{d\mu}(x)\right) dP(x|\theta). 
\end{align}
% \liu{I used the variance definition of Fisher and removed the previous definition as integral of second derivative of log likelihood. The variance definition is closer to the proof of van Trees (via Cauchy Schwarz)}
%Note that we defined the Fisher information in nats (in compliance with the Bayesian Cram\'er-Rao bound discussed in the next), even though the KL divergence is in bits.\uri{I dont get this sentence. The units of FI are $[1/\theta^2]$}

We now record some useful facts concerning the Fisher information. First, we recall that the Fisher information encodes the curvature of the KL divergence w.r.t. translation: Let 
\begin{align}
    g(\theta,\epsilon) \triangleq D\left(P(x|\theta) \| P(x|\theta+\epsilon)\right)
\end{align}
for any $\theta,\theta+\varepsilon\in(a,b)$. The following property is well-known:
\begin{lemma}\label{lem:local_div_fisher}
Under suitable regularity conditions,  $\frac{\partial }{\partial\varepsilon}g(\theta,\varepsilon)|_{\epsilon=0} = 0$, and 
\begin{align}
    \fisher(X; \theta) =\ln2\cdot \left.\frac{\partial^2 g(\theta,\varepsilon)}{\partial\varepsilon^2}\right|_{\varepsilon = 0},
\end{align}
which implies that
\begin{align}\label{eq:kl_to_fisher}
    \fisher(X; \theta) = 2\ln2\cdot \lim_{\varepsilon\to 0}\frac{g(\theta,\varepsilon)}{\varepsilon^2}.
\end{align}
\end{lemma}

\begin{remark}\label{rmrk_regularity}
The ``regularity conditions'' in 
Lemma \ref{lem:local_div_fisher} (and Lemma \ref{lem:bcrlb} below) are to ensure that one can apply the dominated convergence theorem to exchange certain integrals and differentiations in the calculus. 
See for example \cite[Section~2.6]{kullback} for details.
In particular, these conditions are fulfilled if $\sup_{x,\theta}\frac{d P(\cdot|\theta)}{d\mu}(x)<\infty$, $\inf_{x,\theta}\frac{d P(\cdot|\theta)}{d\mu}(x)>0$, and $\sup_{x,\theta}\frac{\partial^m}{\partial\theta^m}\left[\frac{d P(\cdot|\theta)}{d\mu}(x)\right]<\infty$ for $m=1,2,3$.
In the interactive estimation problem, these conditions are always satisfied for sources $
(\bf X,Y)$ on finite alphabets (even if the message alphabets are not finite). Indeed, suppose that $\bf X,Y$ are binary vectors, and that Alice performs an estimation. Let the reference measure $\mu=P^0(\Pi,\bf X)$ be the distribution under $\rho=0$.
We have that 
\begin{align}
\frac{d P^{\rho}}
{d \mu}(\Pi,{\bf x})
=
\frac{\sum_{\bf y}P(\Pi|{\bf x,y})P^{\rho}({\bf x,y})}
{\sum_{\bf y}P(\Pi|{\bf x,y})P^0({\bf x,y})}
\le 
\sup_{\bf x,y}\frac{P^{\rho}({\bf x,y})}{P^0({\bf x,y})}
\end{align}
is bounded by a value independent of $\Pi$.
Similarly, 
\begin{align}
\frac{d P^{\rho}}
{d \mu}(\Pi,{\bf x})
&\ge
\inf_{\bf x,y}\frac{P^{\rho}({\bf x,y})}{P^0({\bf x,y})},
\\
\frac{\partial^m}{\partial \rho^m}\left[\frac{d P^{\rho}}
{d \mu}(\Pi,{\bf x})\right]
&\le 
\sup_{\bf x,y}\frac{\frac{\partial^m}{\partial\rho^m}P^{\rho}({\bf x,y})}{P^0({\bf x,y})}.
\end{align}
\end{remark}

% When $\mathcal{X}$ is a discrete space we can take $\mu$ to be the counting measure on $2^{\mathcal{X}}$, and then $\frac{dP(x|\theta)}{d\mu(x)}$ is simply a conditional probability mass function. In this case, if $\frac{dP(x|\theta)}{d\mu(x)} > 0$ for any $x\in\mathcal{X}, \theta\in(a,b)$ and is infinitely differentiable w.r.t. $\theta\in(a,b)$, then $P(x|\theta)$ is regular, $\fisher(X; \theta) <\infty$, and the relation~\eqref{eq:kl_to_fisher} holds. This simply follows since all the involved quantities are finite \uri{not necessarily!} sums of infinitely differentiable functions.  

% The Fisher information can serve to lower bound quadratic risk of unbiased estimators, i.e., those for which $\EE_\theta \hat\theta = \theta$ for any $\theta \in [a,b]$. 
% \begin{lemma}[Cram\'er-Rao inequality, see e.g.~\cite{van2004detection}]\label{lem:crlb}
% If $P(x|\theta)$ satisfies suitable regularity conditions and $\fisher(X; \theta)<\infty$ for some $\theta$, then
% \begin{align}\label{eq:CRLB}
%     \E_\theta (\hat{\theta}(X) - \theta)^2 \geq \frac{1}{\fisher(X; \theta)}
% \end{align}
% for any unbiased estimator $\hat{\theta}$. 
% \end{lemma}

The Fisher information can be used to lower bound the {\em expected} quadratic risk of estimating $\theta$ from $X$ under a prior distribution on $\theta$. 

% \liu{I removed Cramer Rao (since we do not focus on the unbiased cases)}
\begin{lemma}[Bayesian Cram\'er-Rao inequality, see e.g.~\cite{van2004detection}]\label{lem:bcrlb}
Let $\lambda$ be an absolutely continuous density on a closed interval $\mathcal{J}\subseteq [a,b]$, and assume $\lambda$ vanishes at both endpoints of $\mathcal{J}$. If $P(x|\theta)$ satisfies suitable regularity conditions and $\fisher(X; \theta)<\infty$ for almost all $\theta$,
% \liu{I added ``for each $\theta$''. Didn't check the refs if this is necessary.} then 
\begin{align}\label{eq:bcrlb}
\E_{\theta\sim\lambda}\E_\theta(\hat{\theta}(X) - \theta)^2 \ge \frac{1}{\mathrm{I}^\lambda +  \E_{\theta \sim \lambda} \fisher(X; \theta)}
\end{align} 
for any estimator $\hat{\theta}$, where $\mathrm{I}^\lambda = \int_{\mathcal{J}} \frac{\lambda'^2}{\lambda} d\theta$. 
\end{lemma}
A common choice of prior (see e.g. \cite{tsybakov2009Introduction}) is 
\begin{align}\label{eq:commonPrior}
    \lambda = \frac{2}{|\mathcal{J}|} \lambda_0 \left( \frac{\theta - \theta_0}{|\mathcal{J}|/2} \right)
\end{align}
where $\theta_0$ is the center of the interval $\mathcal{J}$, and $\lambda_0(x) = \cos^2(\pi x/2)$ for $-1 \le x \le 1$ and $0$ otherwise. This prior satisfies $\mathrm{I}^\lambda  = (2 \pi/|\mathcal{J}| )^2$.

% \section{One-way protocols}\label{sec:one-way}
% In this subsection we prove Theorem~\ref{thrm:global} in the special case of one-way protocols. 

\section{Reduction of binary to Gaussian}\label{sec_b2g}
%\liu{A similar argument in LCV}
In this section we show that an achievability scheme for iid Gaussian vector can be converted to a scheme for binary vector by a preprocessing step and applying the central limit theorem (CLT). 
We remark that a similar argument was used in \cite{liu2016common} in the context of common randomness generation.
\begin{lemma}\label{lem_reduction}
Suppose that  $(\Pi,\hat{\rho})$
is a scheme for iid sequence of Gaussian pairs at some length $n$,
and the message alphabet size $|\Pi|<\infty$.
Then there exists a scheme $(\Pi^T,\hat{\rho}^T)$ for iid sequence of binary symmetric pairs of length $T$, for each $T=1,2,\dots$, 
such that 
\begin{align}
\lim_{T\to\infty} H(\Pi^T)&=H(\Pi),\quad
\forall \rho\in[-1,1],
\label{e50}
\\
\lim_{T\to\infty}
R_{\rho}(\Pi^T,\hat{\rho}^T)&\le R_{\rho}(\Pi,\hat{\rho})
,\quad
\forall \rho\in[-1,1],
\label{e51}
\end{align}
where $\rho$ denotes the correlation of the Gaussian or binary pair.
\end{lemma}
\begin{proof}
Let $(A_l,B_l)_{l=1}^t$ be an iid sequence of binary symmetric random variables with correlation $\rho$, and put 
\begin{align}
X^{(t)}&:=\frac{A_1+\dots+A_t}{\sqrt{t}}
+a_tN,\label{e_xt}
\\
Y^{(t)}&:=\frac{B_1+\dots+B_t}{\sqrt{t}}+a_tN',\label{e_yt}
\end{align}
where $N$ and $N'$ are standard Gaussian random variables, and $N$, $N'$, $(X^t,Y^t)$ are independent.
By the central limit theorem, 
we can choose some $a_t=o(1)$ such that
the distribution of $(X^{(t)},Y^{(t)})$ converges to the Gaussian distribution $P_{XY}$ in total variation (Proposition~\ref{prop_clt} below).
Now let $T=nt$ and suppose that $(A_l,B_l)_{l=1}^T$ is an iid  sequence of binary symmetric pairs. 
The above argument shows that Alice and Bob can process locally to obtain iid sequence of length $n$, which convergences to the  iid sequence of Gaussian pairs of correlation $\rho$ in the total variation distance.
After this preprocessing step, Alice and Bob can apply the given scheme $(\Pi,\hat{\rho})$.
Then \eqref{e50} follows since entropy is continuous w.r.t.\ the total variation on finite alphabets, and \eqref{e51} follows since we can assume without loss of generality that $\hat{\rho}$ is bounded.
Note that we have constructed $(\Pi^T,\hat{\rho}^T)$ only for $T$ equal to a multiple of $n$; however this restriction is obviously inconsequential.
\end{proof}
\begin{proposition}\label{prop_clt}
There exist $a_t=o(1)$ such that $X^{(t)}$ and $Y^{(t)}$ defined in \eqref{e_xt} and \eqref{e_yt} converges to the Gaussian distribution $P_{XY}$ in total variation.
\end{proposition}
\begin{proof}
By the convexity of the relative entropy,
we can upper bound the KL divergence by the Wasserstein 2 distance:
\begin{align}
&\quad D(X^{(t)},Y^{(t)}\|X+a_tN,Y+a_tN')
\nonumber
\\
&\le \frac{1}{2a_t^2}
W_2^2\left(\left[\frac{A_1+\dots+A_t}{\sqrt{t}},\frac{B_1+\dots+B_t}{\sqrt{t}}\right],[X,Y]\right)
\label{e_conv}
\end{align}
However, $\frac{A_1+\dots+A_t}{\sqrt{t}}$ and $\frac{B_1+\dots+B_t}{\sqrt{t}}$ converge to $P_{XY}$ under Wasserstein 2 distance,
since this is equivalent to convergence in distribution in the current context where a uniformly integrable condition is satisfied (see e.g.\ \cite[Theorem~7.12]{villani2003topics})
 \footnote{Alternatively, see \cite{murata1974inequality} for a direct proof of the central limit theorem under the Wasserstein metric.}.
Thus there exists $a_t=o(1)$ such that \eqref{e_conv} vanishes.
By Pinsker's inequality, this implies that $(X^{(t)},Y^{(t)})$ converges to the Gaussian distribution $(X+a_tN,Y+a_tN')$ in total variation.
However, as long as $a_t=o(1)$ we have that $(X+a_tN,Y+a_tN')$ converges to $(X,Y)$. 
The conclusion then follows by the triangle inequality of the total variation.
\end{proof}

\section{Proof of the upper bounds (Theorem~\ref{thm_ub})}\label{sec_ub}
Before the proof, let us observe the suboptimality of a naive scheme.
Consider the binary case for example (the Gaussian case is similar). Suppose that Alice just sends her first $k$ samples $X_1,\ldots,X_k$. This would let Bob, by computing the empirical average $\hat \rho_{\textnormal{emp}} = \tfrac{1}{k} \sum_j X_j Y_j$, achieve a risk of
\begin{align}
 \EE_\rho[|\rho - \hat \rho_{\textnormal{emp}}|^2] = \frac{1-\rho^2}{k}.
 \label{e_emp}
\end{align}

Clearly \eqref{e_emp} is not sufficient for the upper bounds in Theorem~\ref{thm_ub}.
To improve it, we now recall the ``max of Gaussian scheme'' in~\cite{hadar2018distributed}.
By a central limit theorem argument we can show that binary estimation is easier than the Gaussian counterpart (see Lemma~\ref{lem_reduction}).
Hence we only need to prove the achievability for the Gaussian case. Alice observes the first $2^k$ Gaussian samples, and transmits to Bob, using exactly $k$ bits, the \emph{index} $W$ of the maximal one, i.e. 
\begin{align}
    W = \argmax_{i \in [2^k]} X_i.
\end{align}
% \uri{$[\cdot]$ wasn't defined. Perhaps we should define it and replace all the $\{1,\ldots,r\}$'s is section 5} 
Upon receiving the index $W$, Bob finds his corresponding sample $Y_W$ and estimates the correlation using
\begin{align}\label{eq:maxEtmtr}
    \hat{\rho}_{\textnormal{max}} = \frac{Y_W}{\EE X_W}.
\end{align}
% In the binary case, the same risk can be achieved, since Alice and Bob can perform local preprocessing to obtain a new iid which converges to Gaussian by the central limit theorem, and then run the Gaussian scheme. 
% The argument is formalized in Lemma~\ref{lem_reduction} ahead.
% % (We remark that a related but different central limit theorem argument for the variable-rate setting was previously used in~\cite[Theorem 5]{hadar2018distributed}).\liu{please verify this} 
% In summary,
Recall the following result \cite{hadar2018distributed}, for which we reproduce the short proof and then explain how the local upper bound will follow with a modification of the proof. 
\begin{theorem}[\cite{hadar2018distributed}]\label{thrm:maxEstmtr}
The estimator $\hat{\rho}_{\textnormal{max}}$ is unbiased with 
\begin{align}\label{eq:maxEstmtrRisk}
     R_\rho\left(W,\hat \rho_{\textnormal{max}}\right) = \frac{1}{k} \left(\frac{1-\rho^2}{2 \ln 2}  + o(1) \right).
\end{align}
% Note that $\EE X_W$ depends only on $k$ and can be calculated at any given accuracy, or be approximated by $\sqrt{k \cdot 2 \ln 2}$
\end{theorem}
\begin{proof}
%[Proof of \eqref{eq:maxEstmtrRisk}]
It is easy to check that $\hat{\rho}_{\textnormal{max}}$ is unbiased. In order to compute its variance, we need to compute the mean and variance of $X_W$, which is the maximum of $2^k$ iid standard normal r.v.s. From extreme value theory (see e.g. \cite{david2004order}) applied to the normal distribution, we obtain
	\begin{align}
	&\E X_W = \sqrt{2 \ln (2^k)}(1 + o(1))\\
	&\E X_W^2 = 2 \ln{(2^k)} (1 + o(1))\\
	&\Var X_W = O \left( \frac{1}{\ln (2^k)} \right).
	\end{align}
Therefore, for $Z \sim \mathcal{N}(0,1)$ we have that 
	 	\begin{align}
\Var \hat{\rho}_{\textnormal{max}} 
&= \frac{1}{(\E X_W)^2} \Var (\rho X_W + \sqrt{1-\rho^2}Z) \\
&= \frac{1}{(\E X_W)^2} (\rho^2 \Var X_W + 1-\rho^2) \\
&= \frac{1}{2k \ln 2} (1-\rho^2 + o(1)).
\end{align}
\end{proof}

Taking $\rho=0$ in \eqref{eq:maxEstmtrRisk} establishes the global upper bound \eqref{eq:glblMinMx}.
However, achieving the local risk upper bound in \eqref{eq:interactive_local_risk_UB} is trickier, since a direct application of \eqref{eq:maxEstmtrRisk} is loose by a factor of $(1-\rho^2)$. 
The trick is to send the index $W$ more efficiently using the side information. More precisely, Alice looks for the maximum sample out of $2^k$ samples as before. Bob sifts his corresponding samples, marking only those where $Y_k >  \rho\cdot \sqrt{k\cdot 2\ln 2}\cdot (1-o(1))$. 
Note that here $\rho$ is as in the definition of the local risk \eqref{e_local}, and the true correlation is within $o(1)$ error to $\rho$. It is easy to check that with sufficiently high probability (sufficiently here (and below) meaning that the complimentary probability has a negligible effect on the ultimate variance), there are $2^{k(1-\rho^2) (1+o(1))}$ such marked samples that also include the one corresponding to Alice's maximum. Also, by symmetry these marked samples are uniformly distributed among the total $2^k$ samples. Hence, Alice can describe the $k\cdot (1-\rho^2)\cdot (1+o(1))$ most significant bits (say) of the index of her maximal sample, which will reveal this index to Bob with sufficiently high probability. This yields a $(1-\rho^2)$ factor saving in communication, and the claim follows. 

\begin{remark}
We note that the above risk can also be achieved directly in Hamming space (without appealing to the CLT). Alice sets some parameter $\tilde{\rho}\in[-1,1]$ to be optimized later, and partitions her data to $m$ blocks of size $n$. She then finds the first block whose sum is exactly $n\tilde{\rho}$ (recall the samples are in $\{-1,1\}$), which exists with sufficiently high probability for $m=2^{n(\frac{1}{2}-h(\frac{1-\tilde{\rho}}{2})+o(1))}$ (otherwise, she picks the first block). Bob sifts his corresponding blocks, marking only those with sum $n\rho\tilde{\rho}(1+o(1))$. Alice encodes the index of her chosen block using $\log m=n(\frac{1}{2}-h(\frac{1-\tilde{\rho}}{2})+o(1))$ bits, and sends only the $n(h(\frac{1-\rho\tilde{\rho}}{2}) - h(\frac{1-\tilde{\rho}}{2}))$ most significant bits, so that Bob can resolve the index with sufficiently high probability. Bob then finds the sum of his corresponding block, and divides it by $n\tilde{\rho}$ to obtain his estimator for $\rho$. It is straightforward to check that this procedure results in a variance of $\frac{1}{k}\cdot \left((1-\rho^2) (h(\frac{1-\rho\tilde{\rho}}{2}) - h(\frac{1-\tilde{\rho}}{2}))/ \tilde{\rho}^2 + o(1)\right)$, where $h(q) = -q\log_2{q}-(1-q)\log_2(1-q)$ is the binary entropy function. Optimizing over $\tilde{\rho}$ yields that $\tilde{\rho}\to 0$ is (not surprisingly) optimal, and the obtained variance is the same as the one achieved by the modified Gaussian maximum estimator above. 
\end{remark}

\section{Proof of the lower bounds in Theorem~\ref{thm_lb} (one-way case)}\label{sec_lb1}
In this section we prove the lower bounds on the global and local one-way risks, $R^{*1}$ and $R^{*1}_{\rho, \delta}$, in the binary case. 
The Gaussian case will then follow from the central limit theorem argument in Lemma~\ref{lem_reduction}.
Of course, the one-way lower bound is a special case of the interactive case in  Section~\ref{sec_interactive};
we separate the discussion simply because the one-way case is conceptually easily and the proof does not need the symmetric strong data processing inequality (modulus certain technical issues pertaining the continuity of Fisher information which we will discuss).

We note that in the one-way setting, the following Markov chain holds:
\begin{align}
    \Pi - \vect{X}-\vect{Y}.
\end{align}
Note that regardless of $\rho$ the marginal distribution of $\vect{X}$ (and thus of $\Pi$) is the same. Let $P_{\Pi \vect{Y}}^{\rho}$ denote the joint distribution of $(\Pi,\vect{Y})$ when the correlation is equal to $\rho$. Note that under $\rho=0$ we have that $\Pi$ and $\vect{Y}$ are independent. Thus, via~\eqref{eq:midkl} we obtain
\begin{align}\label{eq:kl_onew}
     D(P_{\Pi \vect{Y}}^{\rho} \| P_{\Pi \vect{Y}}^{0}) = I(\Pi; \vect{Y})\,.
\end{align}
Furthermore, from~\eqref{eq:sdpi_pre} we get 
\begin{align}\label{eq:oneWaymutInfBound1}
    I(\Pi; \vect{Y}) \le \rho^2 I(\Pi; \vect{X}) \le \rho^2 H(\Pi) \le \rho^2 k.
\end{align}
Thus using the connection between the KL divergence and the Fisher information in
Lemma~\ref{lem:local_div_fisher}, we obtain
\begin{equation}%\label{eq:fishzero}
    \fisher(\Pi,\vect{Y}; \rho=0) \le {k 2\ln 2}\,. 
\end{equation}

Now, suppose that we can show a continuity result for the Fisher information at $\rho=0$, in the sense of 
\begin{align}
    \limsup_{\rho\to0}\sup_{\Pi}\fisher(\Pi,\vect{Y}; \rho) \le {k 2\ln 2}\,
    \label{e_cont}
\end{align}
then a standard application of the Bayesian Cram\'er-Rao bound would imply the global risk.
Indeed, applying Lemma~\ref{lem:bcrlb} with (e.g.) the prior specified in \eqref{eq:commonPrior} over $\mathcal{J} = [\rho - \delta, \rho+\delta]$, we obtain
\begin{align}
    R^{*1} \ge \frac{1}{k} \left( \frac{1}{2 \ln 2} - o(1) \right),
\end{align}
for $\delta \in o(1)\cap\omega(1/\sqrt{k})$,
establishing \eqref{eq:glblMinMx1} for the special case of one-way protocols.

While the continuity claim in \eqref{e_cont} is intuitive enough, to rigorously show it we need to resort to a device to be discussed in Section~\ref{sec_cr}, 
which will allow us to reduce the problem of testing against an arbitrary $\rho$ to testing against independence.
Specifically, in
Corollary~\ref{crlry:anyTwoCorrs} we will show that using common randomness this can be generalized to yield
\begin{align}
    D(P_{\Pi \vect{Y}}^{\rho_1} \| P_{\Pi \vect{Y}}^{\rho_0}) \le \left(\frac{\rho_1-\rho_0}{1-|\rho_0|}\right)^2 k
\end{align}
for any $\rho_1 \in [-1,1]$ and $\rho_0 \in [\tfrac{\rho_1-1}{2} , \tfrac{\rho_1+1}{2}]$. Again applying Lemma \ref{lem:local_div_fisher}, we obtain
\begin{equation}%\label{eq:fishgeneral}
    \fisher(\Pi,\vect{Y}; \rho) \le \frac{2 k \ln 2}{(1 - |\rho|)^2}
\end{equation}
for any $\rho \in (-1,1)$.
This justifies the continuity claim \eqref{e_cont}.
Moreover,
applying the Bayesian Cram\'er-Rao (Lemma~\ref{lem:bcrlb}) with (e.g.) the prior specified in \eqref{eq:commonPrior} over $\mathcal{J} = [\rho - \delta, \rho+\delta]$, we obtain
\begin{align}
    R^{*1}_{\rho, \delta} \ge \frac{1}{k} \left( \frac{(1 - |\rho|)^2}{2 \ln 2} - o(1) \right)
\end{align}
which is the desired local risk lower bound for the special case of one-way protocols.

%Next, we aim to apply the lower bound~\eqref{eq:bcrlb}. Let $\mathcal{J}=[-\delta,\delta]$. The minimax risk over $\mathcal{J}$ is clearly lower bounded by any Bayesian risk with support on $\mathcal{J}$, i.e.,
%\begin{align}
%	\inf_{\hat{\rho}} \sup_{\rho \in \mathcal{J}} R_\rho(\hat{\rho}) & = \inf_{\hat{\rho}} \sup_{\lambda:\mathrm{supp}(\lambda) \subseteq \mathcal{J}} \E_{\rho\sim\lambda}R_\rho(\hat{\rho})\\
%	&\geq \sup_{\lambda:\mathrm{supp}(\lambda) \subseteq \mathcal{J}} \inf_{\hat{\rho}} \E_{\rho\sim\lambda}R_\rho(\hat{\rho})
%\end{align}
%We select the prior as
%	\begin{align}\label{eq:prior}
%		\lambda(\rho) = \frac{15}{16}\frac{(\delta^2 - \rho^2)^2}{\delta^5}\cdot\mathds{1}(|\rho|\leq \delta), 
%	\end{align}
%	which satisfies 
%	\begin{align}\label{eq:priorFish}
%	\mathrm{I}^\lambda = 10/\delta^2.
%	\end{align}
%	Plugging into~\eqref{eq:bcrlb} and recalling \eqref{eq:fishzero}, we have that 
%	\begin{align}
%		\inf_{\hat{\rho}}\sup_{\rho} R_\rho(\hat{\rho}) &\geq \inf_{\hat{\rho}}\sup_{|\rho|\leq \delta} R_\rho(\hat{\rho}) \\
%		&\geq \inf_{\hat{\rho}}\E_{\rho\sim\lambda} R_\rho(\hat{\rho}) \\
%		&\geq \left(k\cdot 2\ln 2 + \max_{|\rho|\leq \delta}|I(\rho)-I(0)| + \frac{10}{\delta^2}\right)^{-1}.
%	\end{align}
%	The required lower bound now follows by taking, say, $\delta=k^{-1/4}$, and using the continuity of the Fisher information \uri{explain why this holds}. 

% In Section~\ref{sec:interact} the ideas above will be extended to interactive communication. 

\section{Proof of lower bounds in Theorem~\ref{thm_lb} (interactive case)}\label{sec_interactive}
For the interactive case, our approach is again to upper bound the KL divergence between the distributions of the r.v.s available (to either Alice or Bob) under $\rho\neq 0$, and under $\rho=0$. 
This is accomplished by Theorem~\ref{thm5} and Theorem~\ref{thm6} below, which can be viewed as generalizations of \eqref{eq:kl_onew} and \eqref{eq:oneWaymutInfBound1}.

\def\vX{\vect{X}}
\def\vY{\vect{Y}}
\def\vZ{\vect{Z}}
\begin{theorem}\label{thm5} Consider an arbitrary interactive protocol $P_{\Pi|\vX \vY}$ and let $P_{\vX \vY \Pi}$ be the induced joint distribution. Let $\bar P_{\vX \vY \Pi} = P_{\vX} \times P_{\vY} \times P_{\Pi|\vX \vY}$ be the joint distribution induced by the same protocol, but when the $\vX$ and $\vY$ are taken to be independent (but with same marginals). Then 
\begin{align}\label{eq:dtoic}
\max\{D(P_{\Pi  \vX}\|\bar P_{\Pi  \vX}),\,
D(P_{\Pi \vY}\| \bar P_{\Pi \vY})\}
\le 
I(\vX;\vY)-I(\vX;\vY|\Pi), 
\end{align}
where information quantities are computed with respect to $P_{\vX \vY \Pi}$.
Moreover, the bound \eqref{eq:dtoic} continues to hold also when the protocol $\Pi$ contains an arbitrary common randomness (i.e.\ public coin) $W_0$ independent of $(\vX,\vY)$.
\end{theorem}
% \begin{theorem}\label{thm5}
% Consider an arbitrary source pair $(\vX,\vY)$ and common randomness $\vZ$ independent of the source pair.
% Let $P_{\Pi|\vX \vY \vZ}$ be an arbitrary interactive protocol which may use the common randomness. 
% Let $P_{\vX \vY \vZ \Pi}=P_{\vZ}P_{\vX \vY}P_{\Pi|\vX \vY \vZ}$ be the induced joint distribution. 
% Let $\bar P_{\vX,\vY,\Pi} = P_{\vX}P_{\vY}P_{\vZ}P_{\Pi|\vX,\vY,\vZ}$ be the joint distribution induced by the same protocol, but $\vX$, $\vY$ and $\vZ$ are independent (but with same marginals). Then 
% \begin{align}\label{eq:dtoic}
% \max\{D(P_{\Pi,\vX,\vZ}\|
% \bar{P}_{\Pi,\vX,\vZ}),\,
% D(P_{\Pi,\vY,\vZ}\|
% \bar{P}_{\Pi,\vY,\vZ})
% \}
% \le 
% I(\vX;\vY)-I(\vX;\vY|\Pi), 
% \end{align}
% where (conditional) mutual information quantities are computed with respect to $P_{\vX,\vY,\vZ,\Pi}$.
% \end{theorem}

By saying that the protocol contains common randomness we mean that $\Pi=(W_0,W_1,\dots,W_r)$ where $W_0$ is the common randomness and $W_1,\dots,W_r$ are the exchanged messages. 
The extension to the case of common randomness will be useful in Section~\ref{sec_cr} where we reduce the problem of testing against an arbitrary $\rho$ to testing against independence.
\begin{proof}[Proof of Theorem~\ref{thm5}]
First, since 
\begin{align}
D(P_{\Pi\vX W_0}\|\bar{P}_{\Pi\vX W_0})=D(P_{\Pi\vX|W_0}\|\bar{P}_{\Pi\vX|W_0}|P_{W_0}),
\end{align}
it suffices to prove the same upper bound for 
\begin{align}
    D(P_{\Pi\vX|W_0=w_0}\|\bar{P}_{\Pi\vX|W_0=w_0}).
\end{align}
In other words, it suffices to prove the theorem for the case where the common randomness $W_0$ is empty.
Under this assumption, note that the RHS of~\eqref{eq:dtoic} is equal to
\begin{align}
    & I(\vX;\vY)-I(\vX;\vY|\Pi) \nonumber\\
    {} &= I(\vX; \Pi) + I(\vY; \Pi) - I(\vX, \vY; \Pi)\\
                &= \EE\left[ \log {P_{\vX  \Pi}(\vX,\Pi)  P_{\vY  \Pi}(\vY,\Pi) P_{\vX \vY}(\vX,\vY) P_\Pi(\Pi) \over P_{\vX}(\vX) P_\Pi(\Pi) P_{\vY}(\vY) P_\Pi(\Pi) P_{\vX  \vY  \Pi}(\vX, \vY, \Pi)} \right]\\
                & = \EE\left[ \log {P_{\vX  \Pi}(\vX,\Pi)  P_{\vY|\Pi}(\vY|\Pi)
                \over \bar P_{\vX  \vY  \Pi}(\vX,\vY,\Pi)} \right]\\
                & = \EE\left[ \log {P_{\vX  \Pi}(\vX,\Pi)  \over \bar P_{\vX  \Pi}(\vX,\Pi)} + \log {P_{\vY|\Pi}(\vY|\Pi)
                \over  \bar P_{\vY|\vX \Pi}(\vY|\vX,\Pi)} \right]\label{eq:dtoic2}\\
                & = D(P_{\vX  \Pi} \| \bar P_{\vX  \Pi}) + \EE\left[ \log {P_{\vY|\Pi}(\vY|\Pi)
                \over \bar P_{\vY|\vX \Pi}(\vY|\vX,\Pi)} \right]\\
                & \ge D(P_{\vX \Pi} \| \bar P_{\vX  \Pi})\,.\label{eq:dtoic1}
\end{align}
where all expectations are taken with respect to $P_{\vX \vY \Pi}$ and the last step is by non-negativity of divergence
 $D(P_{\vY|\Pi=\pi} \| \bar P_{\vY|\Pi=\pi, \vX = \vect{x}})$
for all $\pi,\vect{x}$,
which in turn uses the Markov chain $\vX-\Pi-\vY$ under $\bar{P}_{\vX\Pi\vY}$. In all,~\eqref{eq:dtoic1} proves part of~\eqref{eq:dtoic}. To prove the same bound on $D(P_{\Pi \vY}\| \bar P_{\Pi \vY})$ we can argue by symmetry (it may seem that symmetry is broken by the fact that $\vX$ sends $W_1$ first, but this is not true: $W_1$ can be empty), or just perform a straightforward modification of step~\eqref{eq:dtoic2}.
% \iffalse
% Given a protocol for computing $W_1,\dots,W_r$, where $r\in 2\mathbb{Z}$.
% Let $P_{X^nY^nW^r}$ be the true distribution of the random variables.
% Then construct the following distributions
% \begin{align}
% P^{(r)}&:=P_{Y^nW^r}\cdot P_{X^n|W^r}
% \\
% P^{(r-1)}&:= P_{X^nW^{r-1}}\cdot P_{Y^n|W^{r-1}}\cdot 
% P_{W_r|X^nY^nW^{r-1}}
% \\
% P^{(r-2)}&:= P_{Y^nW^{r-2}}\cdot P_{X^n|W^{r-2}}\cdot 
% P_{W_{r-1}^r|X^nY^nW^{r-2}}
% \\
% &\dots
% \\
% P^{(1)}&:= P_{X^nW_1}\cdot P_{Y^n|W_1}\cdot P_{W_2|X^nY^nW_2^r}
% \\
% P^{(0)}&:= P_{X^n}\cdot P_{Y^n} \cdot P_{W^r|X^nY^n}.
% \end{align}
% Then
% \begin{align}
% D(P_{X^nW^r}\|P^{(0)}_{X^nW^r})
% &=D(P^{(r)}_{X^nW^r}\| P^{(0)}_{X^nW^r})
%  \\
% &\le D(P^{(r)}\| P^{(0)})
% \\
% &= D(P^{(r)}\| P^{(1)}) 
% + \mathbb{E}_{P^{(r)}}\left[\log\frac{{\rm d}P^{(1)}}{{\rm d}P^{(0)}} \right]
% \\
% &= D(P^{(r)}\| P^{(1)}) 
% + \mathbb{E}_{P^{(r)}}\left[\log\frac{{\rm d}P^{(1)}_{Y^nW_1}}{{\rm d}P^{(0)}_{Y^nW_1}} \right]
% \\
% &=D(P^{(r)}\| P^{(1)}) 
% + \mathbb{E}_{P^{(1)}}\left[\log\frac{{\rm d}P^{(1)}_{Y^nW_1}}{{\rm d}P^{(0)}_{Y^nW_1}} \right]
% \\
% &=D(P^{(r)}\| P^{(1)}) + I(W_1;Y^n)
% \\
% &=D(P^{(r)}\| P^{(2)}) + I(W_2;X^n|W_1)+ I(W_1;Y^n)
% \\
% &\dots
% \\
% &=I(W_r;X^n|W^{r-1})+ \dots + I(W_2;X^n|W_1)+I(W_1;Y^n).
% \label{e18}
% \end{align}
% where
% \begin{itemize}
%  \item When we write $P^{(0)}$, etc., without any subscript, we mean the distribution of all the variables $X^nY^nW^r$.
%  \item The last two steps follow by induction.
% \end{itemize}
% \fi
\end{proof}

\begin{remark} It can be seen that for a one-way protocol we have equality in~\eqref{eq:dtoic}. This explains why our impossibility bound can be essentially achieved by a one-way protocol (e.g., see Theorem~\ref{thrm:maxEstmtr}), and suggests that this is the only possibility.
\end{remark}

\begin{remark}
After completion of this work, we found out that in a slightly different form Theorem~\ref{thm5} has previously appeared in \cite[Equation (4)]{xiang2013interactive}. Our proof is slightly simpler.
\end{remark}

\begin{theorem}\label{thm6}
Let $\Pi$ be any interactive protocol, possibly containing a common randomness $W_0$, in either the Gaussian or the binary symmetric case. Then    
\begin{align}
I(\mathbf{X;Y})-I(\mathbf{X;Y}|\Pi)   
\le
\rho^2 I(\Pi; \mathbf{X,Y}). 
\label{e67}
\end{align}
\end{theorem}

The proof of Theorem~\ref{thm6} is given in Section~\ref{sec_proof_ssdpi}.

\begin{remark}
%Recall that $\Pi$ denotes the collection of all messages $W^r$.
The following notions of external and internal information costs were introduced in \cite{CSWY01} and \cite{BBCR10} respectively:
\begin{align}
{\sf IC^{ext}}_P(\Pi)&:=I(\Pi;\mathbf{X,Y});
\\
{\sf IC}_P(\Pi)&:={\sf IC^{ext}}_P(\Pi)-[I(\mathbf{X;Y})-I(\mathbf{X;Y}|\Pi)].
\end{align}
Using the Markov chain conditions of the messages
\begin{align}
&W_i-(\mathbf{X},W^{i-1})-\mathbf{Y},\quad i\in [r]\setminus2\mathbb{Z},
\\
&W_i-(\mathbf{Y},W^{i-1})-\mathbf{X},\quad i\in [r]\cap2\mathbb{Z}
\end{align}
we will be able to write the external and internal information as sums of information gains in each round of communication:
\begin{align}
I(\mathbf{X;Y})-I(\mathbf{X;Y}|\Pi)
=\!\! \sum_{i\in[r]\setminus 2\mathbb{Z}} \!\! I(W_i; \mathbf{Y}|W^{i-1}) +
\!\! \sum_{i\in [r]\cap
2\mathbb{Z}} \!\! I(W_i; \mathbf{X}|W^{i-1}),
\label{eq71}
\\
I(\Pi;\mathbf{X,Y})= \!\! \sum_{i\in[r]\setminus 2\mathbb{Z}} \!\! I(W_i; \mathbf{X}|W^{i-1}) 
+
\!\! \sum_{i\in [r]\cap2\mathbb{Z}} \!\! I(W_i; \mathbf{Y}|W^{i-1}),
\end{align}
which are useful later in some proofs.
\end{remark}

\section{Reduction of testing against arbitrary $\rho$ to testing against independence}\label{sec_cr}
The results in Theorem~\ref{thm5} and Theorem~\ref{thm6} only (directly) applies to testing against independence, and hence are insufficient for handling the local risks at an arbitrary $\rho$.
Fortunately, for binary and Gaussian vectors, there is a simple device of translating the correlations by leveraging the common randomness,
so that the general problem is reduced to the case of testing independence solved in Theorem~\ref{thm5} and Theorem~\ref{thm6}.
More precisely, we obtained the following result:
\begin{corollary}\label{crlry:anyTwoCorrs}
%Consider arbitrary $\rho_1\in[-1,1]$ and $\rho_0 \in [\tfrac{\rho_1 - 1}{2} , \tfrac{\rho_1 + 1}{2}]$.
Let $P_{\vX \vY}^{\rho_0}$ (resp.\ $P_{\vX \vY}^{\rho_1}$) be the joint distribution for Gaussian or binary symmetric vector sources under correlation $\rho_0$ (resp.\ $\rho_1$).
Let $P_{\Pi|\vX \vY}$ be an arbitrary protocol. 
Then for any $\rho_1\in[-1,1]$ and $\rho_0 \in [\tfrac{\rho_1 - 1}{2} , \tfrac{\rho_1 + 1}{2}]$, 
\begin{align}\label{eq:klbndTwocorrs}
\max\{D(P_{\Pi \vX}^{\rho_1}\|P_{\Pi \vX}^{\rho_0}),\,
D(P_{\Pi \vY}^{\rho_1}\|P_{\Pi \vY}^{\rho_0})\}
\le 
\left(\frac{\rho_1-\rho_0}{1-|\rho_0 |}\right)^2 k.
\end{align}
In particular, this bounds the Fisher information in the case of finite-length binary vectors as
\begin{align}\label{eq:fishBndIntrctve}
    \max \{\fisher(\Pi,\vX; \rho) , \fisher(\Pi,\vY; \rho) \}\le \frac{2 k \ln 2}{(1 - |\rho|)^2}.
\end{align}
\end{corollary}
% Remark: RW did not cover the interactive case; besides this, this theorem is the same as RW \uri{modify}.
\begin{proof}
From Theorems~\ref{thm5}~and~\ref{thm6} we have 
\begin{align}
\max\{D(P_{\Pi \vX}^{\rho}\|P_{\Pi \vX}^{0}),\,
D(P_{\Pi \vY}^{\rho}\|P_{\Pi \vY}^{0})\}
\le 
\rho^2 I(\Pi;\vX,\vY) \le \rho^2 k.
\end{align}
% for any $\rho$, for both the Gaussian and binary cases. We now introduce common randomness to each coordinate in order to shift the correlation values as follows.
The proof uses a device of shifting the correlation by introducing common randomness.
Suppose that $\vX$ and $\vY$ are iid binary or Gaussian vectors of length $n$, where the correlation between $X_i$ and $Y_i$ is $0$ under $P^{(0)}$ and $\rho:=\frac{\rho_1-\rho_0}{1-|\rho_0|}$ under $P^{(1)}$, for each $i\in [n]$.
We define the common randomness $W_0$ independent of $\vX$ and $\vY$ as follows:
\begin{itemize}
    \item Gaussian case:  Let $W_0=\vZ$, where $Z_i \sim \mathcal{N}(0,1)$ are iid, and define
    \begin{align}
        X'_i &= \alpha Z_i + \sqrt{1 - \alpha^2} X_i \\
        Y'_i &= s \alpha Z_i + \sqrt{1 - \alpha^2} Y_i
    \end{align}
    for some $\alpha \in [-1,1]$ and $s \in \{-1, 1\}$.
    \item Binary case: Let $W_0  = ({\bf B,Z})$ where ${\bf B}$ is independent of ${\bf Z}$, $B_i \sim \Ber{(\alpha)}$ over $\{0,1\}$ are iid and $Z_i \sim \Ber{(\tfrac{1}{2})}$ over $\{-1,1\}$ are iid. 
    Put
    \begin{align}
        X'_i &= B_iZ_i + (1-B_i)X_i\\
        Y'_i &= sB_iZ_i + (1-B_i)Y_i
    \end{align}
    for some $\alpha \in [0,1]$ and $s \in \{-1, 1\}$.
\end{itemize}
In both cases,
it can be verified that by appropriately choosing $s$ and $\alpha$, the correlation between $X'_i$ and $Y'_i$ equals $\rho_0$ under $P^{(0)}$ and $\rho_1$ under $P^{(1)}$.
Now, consider any protocol $\Pi=(W_0,W^r)$ for the source ${\bf X,Y}$ which includes the common randomness $W_0$.
We have
% \begin{align}
%     D(P_{\Pi,\vY'}^{\rho_1}\|P_{\Pi,\vY'}^{\rho_0}) &= \E_{W_0} D(P_{\Pi,\vY'|W_0}^{\rho_1}\|P_{\Pi,\vY'|W_0}^{\rho_0}) \\ \label{eq:transition63}
%     &= \E_{W_0} D(P_{\Pi,\vY|W_0}^{\rho}\|P_{\Pi,\vY|W_0}^{0}) \\
%     &= D(P_{\Pi,\vY}^{\rho}\|P_{\Pi,\vY}^{0}) \\
%     &\le \rho^2 k.
% \end{align}
\begin{align}
D(P^{(1)}_{W_0W^r\vY'}\|P^{(0)}_{W_0W^r\vY'})    
&\le
D(P^{(1)}_{W_0W^r\vY}\|P^{(0)}_{W_0W^r\vY}) 
\label{eq:transition63}\\
&\le \rho^2 I(W_0,W^r;\vX,\vY)
\label{e_thm56}
\\
&\le \rho^2 I(W^r;\vX,\vY|W_0)
\\
&\le \rho^2k
\end{align}
where \eqref{eq:transition63} follows since $P^{(1)}_{\vY'|W_0W^r\vY}=P^{(0)}_{\vY'|W_0W^r\vY}=P_{\vY'|W_0\vY}$ (note $\vY'$ is a (deterministic) function of $(W_0,\vY)$), and
%\uri{\eqref{eq:transition63} is not accurate in the binary case. To be resolved}\liu{should be fixed now. Please double check.}
\eqref{e_thm56} follows from Theorem~\ref{thm5} and Theorem~\ref{thm6}.
Observe that $D(P^{(1)}_{W_0W^r\vY'}\|P^{(0)}_{W_0W^r\vY'})$ is exactly $D(P_{\Pi \vY}^{\rho_1}\|P_{\Pi \vY}^{\rho_0})$ which we wanted to upper bound.  
Repeating the same steps for $D(P_{\Pi \vX}^{\rho_1}\|P_{\Pi \vX}^{\rho_0})$ establishes \eqref{eq:klbndTwocorrs} for both the Gaussian and binary cases.

The bound on the Fisher information in the binary case \eqref{eq:fishBndIntrctve} follows from Lemma~\ref{lem:local_div_fisher}.
% the arguments mentioned right after \eqref{eq:oneWaymutInfBound1}.
\end{proof}

\begin{remark}\label{rmrk_gaussianishard}
While we expect that the same bound in \eqref{eq:fishBndIntrctve} continues to hold in the Gaussian case, 
the regularization condition required in
the transition from the KL divergence bound to the Fisher information bound appears difficult to justify in the Gaussian case (see Lemma~\ref{lem:local_div_fisher} and the ensuing remark).
\end{remark}

\section{Proof of the symmetric strong data processing inequality}\label{sec_proof_ssdpi}
This section proves  Theorem~\ref{thm6}, which states that the symmetric strong data processing inequality constant is bounded by $\rho^2$ in the case of binary symmetric or Gaussian vectors. 
We first outline the proof, and then supplement the key lemmas used.
\begin{proof}[Proof of Theorem~\ref{thm6}]
First, note that we only need to prove the case where the common randomness $W_0$ is empty.
Indeed, since $\Pi$ includes $W_0$ and since $W_0$ is independent of $(\vX,
\vY)$, we have
$I(\mathbf{X;Y}|\Pi)=I(\mathbf{X;Y}|W_0,\Pi)$
and 
$I(\Pi;\mathbf{X,Y})=I(\Pi;\mathbf{X,Y}|W_0)$, hence \eqref{e67} will follow if we establish
\begin{align}
I(\mathbf{X;Y})-I(\mathbf{X;Y}|\Pi,W_0=w_0)   
\le
\rho^2 I(\Pi; \mathbf{X,Y}|W_0=w_0). 
\end{align}
for each $w_0$.
Using the Markov chains satisfied by the messages we have
\begin{align}
I(\mathbf{X;Y})-I(\mathbf{X;Y}|W^r)
&= \!\!\! \sum_{i\in [r]\cap2\mathbb{Z}} \!\!\! I(W_i;\mathbf{X}|W^{i-1})
+ \!\!\! \sum_{i\in [r]\setminus
2\mathbb{Z}} \!\!\! I(W_i;\mathbf{Y}|W^{i-1})
\label{e79}
\\
I(W^r;\mathbf{X,Y})
&= \!\!\!  \sum_{i\in [r]\cap2\mathbb{Z}} \!\!\! I(W_i;\mathbf{Y}|W^{i-1})
+ \!\!\! \sum_{i\in [r]\setminus
2\mathbb{Z}} \!\!\! I(W_i;\mathbf{X}|W^{i-1}).
\label{e19}
\end{align}
Then the result for the binary and Gaussian cases follow respectively from Lemma~\ref{lem3} and Lemma~\ref{lem11},
as well as the tensorization property Lemma~\ref{lem4}, stated and proved
below.
\end{proof}

\subsection{Binary case}\label{sec_binary}

Our goal is to prove Lemma~\ref{lem3}, which follows from Lemma~\ref{lem_symmetric} stated earlier and proved in this section.

\begin{lemma}[{\cite{liu2016common}}]\label{lem3} Let $X,Y\in\{1,-1\}$ be equiprobably distributed with correlation $\rho$.
Consider any random variables $U^r$, $r\in \mathbb{Z}$,
satisfying 
\begin{align}
&U_i-(X,U^{i-1})-Y,\quad i\in [r]\setminus2\mathbb{Z},
% \label{e21} not referenced, also theres another label with the same name
\\
&U_i-(Y,U^{i-1})-X,\quad i\in [r]\cap2\mathbb{Z}.
\end{align}
Define 
\begin{align}
R(P_{U^rXY}) &:= \sum_{i\in [r]\cap2\mathbb{Z}}I(U_i;X|U^{i-1})
+\sum_{i\in [r]\setminus
2\mathbb{Z}}I(U_i;Y|U^{i-1});
\label{eq:R}
\\
S(P_{U^rXY}) &:= \sum_{i\in [r]\cap2\mathbb{Z}}I(U_i;Y|U^{i-1})
+\sum_{i\in [r]\setminus
2\mathbb{Z}}I(U_i;X|U^{i-1}).
\label{eq:S}
\end{align}
Then
$R(P_{U^rXY}) \le \rho^2 S(P_{U^rXY})$.    
\end{lemma}

\begin{proof}
It suffices to show that the ratio of the $i$-th term on the right side of \eqref{eq:R} to the $i$-th term on the right side of \eqref{eq:S} is upper-bounded by $\rho^2$ for any $i$.
Consider without loss of generality any $i\in2\mathbb{Z}$.
Note that by inducting on $i$ and using the Markov chain conditions satisfied by $U^r$, we observe that $P_{YX|U^{i-1}=u^{i-1}}$ has the property that
\begin{align}
\frac{{\rm d}P_{YX|U^{i-1}=u^{i-1}}}{{\rm  d}P_{XY}}
=f(x)g(y),\quad
\forall x,y
\end{align}
for some functions $f$ and $g$.
Then using Lemma~\ref{lem_symmetric} we conclude that for each $u^i$ we have $\frac{I(U_i;X|U^{i-1}=u^{i-1})}{I(U_i;Y|U^{i-1}=u^{i-1})}\le \rho^2$. 
\end{proof}

The following result is used in the proof of Lemma~\ref{lem_symmetric}. 
We state it in the general vector case, though we only need the scalar $(X,Y)$ case.
\begin{lemma}\label{lem:sdpi} Let $X,Y$ be binary $\PP[X=1]=1-\PP[X=0] = p$ and let $\PP[Y\neq X|X] = {1-\rho\over2}$, $\rho \in [-1,1]$.
Consider $(\vect{X},\vect{Y})$ to be $n$ iid copies of $(X,Y)$. Then for any random variables $U,V$ such that $U-\vect{X}-\vect{Y}-V$ we have
\begin{align}
     I(U;\vect{Y}) &\le \rho^2 I(U;\vect{X})\label{eq:sdpi_pre}\\
     I(\vect{X};V) &\le \rho^2 I(\vect{Y};V)\label{eq:sdpi_post}.
\end{align}
\end{lemma}
\begin{proof}
We first recall that a result known as tensorization (due to~\cite{AG76} in this context) allows to only check $n=1$ case. For $n=1$, the first part~\eqref{eq:sdpi_pre} is the standard inequality dating back to~\cite{AG76}, see~\cite{polyanskiy2017strong} for a survey. To show inequality~\eqref{eq:sdpi_post}, we apply Theorem 21 in~\cite{polyanskiy2017strong}, which establishes the following. Let $A$ be a binary input and $B\sim P$ when $A=0$ and $B\sim Q$  when $A=1$, where $P=(P(v),v=0,1,\ldots)$ and $Q=(Q(v), v=0,1,\ldots)$ are two arbitrary distributions. Then for any $U-A-B$ we have
\begin{equation}
I(U;B) \le I(U;A) \left(1- \left( \sum_v \sqrt{P(v) Q(v)} \right)^2\right)\,. \label{eq:eta_bin}
\end{equation}
(The bound is tight, cf.~\cite[Remark 8]{polyanskiy2017strong}, whenever $B$ is binary.) Applying this result to $A=Y$ and $B=X$ and denoting $q=p\rho + {1-\rho\over2}$ we get
$$ \sum_v \sqrt{P(v) Q(v)} = {\sqrt{1-\rho^2}\over 2 \sqrt{q(1-q)}} \ge \sqrt{1-\rho^2}\,.$$
\end{proof}
\begin{proof}[Proof of Lemma~\ref{lem_symmetric}] Due to symmetry, it suffices to prove only the first inequality. Computing $P_{Y|X}$ and applying~\eqref{eq:eta_bin} we need to prove
$$
\sum_{y\in\{0,1\}} \sqrt{Q_{Y|X}(y|0)Q_{Y|X}(y|1)} { g(y) \over \sqrt{g_0 g_1}} \ge \sqrt{1-\rho^2}\,, 
$$ 
where $g_x = \sum_{y'} g(y') Q_{Y|X}(y'|x), x\in \{0,1\}$. Note that for all $y$,
\begin{align}
\sqrt{Q_{Y|X}(y|0)Q_{Y|X}(y|1)} = \sqrt{1-\rho^2\over 4}. 
\end{align}
By rescaling $g$ so that $\sum_y g(y) = 1$ we get that $g_0+g_1=1$ and hence $\sqrt{g_0 g_1} \le {1\over 2}$, as required.
\end{proof}

\subsection{Tensorization}
The bound in Lemma~\ref{lem_symmetric} does not tensorize. That is, if $Q_{XY}$ in the lemma is replaced by $Q_{XY}^{\otimes n}$,
then $\sup_{P_{UXY}}\frac{I(U;Y^n)}{I(U;X^n)}$ can be strictly larger than $\rho^2$.
Thus the cases of binary symmetric and Gaussian vectors cannot be proved via Lemma~\ref{lem_symmetric} as in the case of a pair of binary variables.
This is a subtle issue that makes the proof of Theorem~\ref{thm6} somewhat nontrivial.
% , since we cannot simply bound the ratio of the corresponding summands in \eqref{e79} and \eqref{e19}.
% Instead, we need to first perform some sophisticated tensorizations for \eqref{e79} and \eqref{e19} where the summands interact.
% Next, we argue that the ratio $\frac{R}{S}$ in Lemma~\ref{lem3} tensorizes.
Luckily, the symmetric strong data procesing constant tensorizes:
\begin{lemma}\label{lem4}
%Let $P_{XY}$ be the distribution 
%of either the symmetric binary random variables or Gaussian random variables with correlation $\rho$.
%under which $X,Y\in\{1,-1\}$ are both equiprobably distributed and $\rho:=\mathbb{E}[XY]$.
Let $(\mathbf{X,Y}):=(X_j,Y_j)_{j=1}^n\sim \otimes_{j=1}^n P_{X_jY_j}$ for any given $P_{X_jY_j}$, $j=1,\dots n$.
Consider any random variables  $W^r$,
%$|\mathcal{W}^r|<\infty$,
$r\in \mathbb{Z}$,
satisfying 
\begin{align}
&W_i-(\mathbf{X},W^{i-1})-\mathbf{Y},\quad i\in [r]\setminus2\mathbb{Z},
\label{e21}
\\
&W_i-(\mathbf{Y},W^{i-1})-\mathbf{X},\quad i\in [r]\cap2\mathbb{Z}.
\end{align}
% Let
% \begin{align}
% R(P_{\mathbf{XY}W^r}) &:= \sum_{i\in[r]\cap2\mathbb{Z}}I(\mathbf{X}; W_i|W^{i-1}) +
% \sum_{i\in[r]\setminus2\mathbb{Z}}I(\mathbf{Y}; W_i|W^{i-1});
% \label{eq:R}
% \\
% S(P_{\mathbf{XY}W^r}) &:= \sum_{i\in[r]\cap2\mathbb{Z}}I(\mathbf{Y}; W_i|W^{i-1}) + \sum_{i\in[r]\setminus2\mathbb{Z}}I(\mathbf{X}; W_i|W^{i-1}).
% \label{eq:S}
% \end{align}
Then
\begin{align}
\frac{R(P_{\mathbf{XY}W^r})}{S(P_{\mathbf{XY}W^r})}\le\max_{1\le i\le n}\sup_{P_{U^r|X_jY_j}}\frac{R(P_{X_jY_jU^r})}{S(P_{X_jY_jU^r})} 
\label{e91}
\end{align}
where $P_{U^r|X_jY_j}$ is such that
\begin{align}
&U_i-(X_j,U^{i-1})-Y_j,\quad i\in [r]\setminus2\mathbb{Z},
\\
&U_i-(Y_j,U^{i-1})-X_j,\quad i\in [r]\cap2\mathbb{Z}.
\end{align}    
\end{lemma}
\begin{proof}
% \yp{JINGBO: This will not do. This lemma is the culmination of all the paper, it has to have a complete and impeccably clean proof. What I mean is $r=2, n=2$ is not good. The $n=2$ is easy to fix: just say that Lemma 9 establishes base case and here we show the induction step (so only deal with $n=2$). $r=2$ should be handled explicitly.}
Note that by induction it suffices to consider $n=2$.
Define 
\begin{align}
U_i&:=(W_i,Y_2),\quad i=1,2,\dots, r;
\\
\bar{U}_i&:=(W_i,X_1), \quad i=1,2,\dots,r.
\end{align}
Then note that the Markov chains
\begin{align}
&U_i-(U^{i-1},X_1)-Y_1,
\quad i\in [r]\setminus2\mathbb{Z},
\\
&\bar{U}_i-(\bar{U}^{i-1},X_2)-Y_2,
\quad i\in
[r]\setminus2\mathbb{Z},
\\
&U_i-(U^{i-1},Y_1)-X_1,
\quad i\in
[r]\cap2\mathbb{Z},
\\
&\bar{U}_i-(\bar{U}^{i-1},Y_2)-X_2,
\quad i\in [r]\cap2\mathbb{Z},
\end{align}
are satisfied.
Moreover,
\begin{align}
&\quad R(P_{W^rX^2Y^2})
\nonumber\\
&= \sum_{i\in[r]\setminus 2\mathbb{Z}}
[I(W_i;Y_2|W^{i-1})+I(W_i;Y_1|W^{i-1},Y_2)]
\nonumber\\
&\quad+\sum_{i\in[r]\cap2\mathbb{Z}}[I(W_i;X_2|W^{i-1},X_1)
+I(W_i;X_1|W^{i-1})]
\\
&= \sum_{i\in[r]\setminus 2\mathbb{Z}}
[I(W_i,X_1;Y_2|W^{i-1},X_1)-\Delta_i
+I(W_i,Y_2;Y_1|W^{i-1},Y_2)]
\nonumber\\
&\quad+ \!\! \sum_{i\in[r]\cap2\mathbb{Z}} \!\! [I(W_i,X_1;X_2|W^{i-1},X_1)+I(W_i,Y_2;X_1|W^{i-1},Y_2)-\Delta_i]
\\
 &=R(P_{U^rX_1Y_1})+R(P_{\bar{U}^rX_2Y_2})  -\sum_{i=1}^r\Delta_i
  \label{e27}
 \\
&= R(P_{U^rX_1Y_1})+R(P_{\bar{U}^rX_2Y_2})  -I(X_1;Y_2|W^r)
 \label{e28}
\end{align}
where we have defined $\Delta_i:=I(X_1;Y_2|W^i)-I(X_1;Y_2|W^i)$, and in
\eqref{e27} we have used the independence $Y_1\perp Y_2$ for the $i=1$ base case.
Next, with similar algebra we obtain
\begin{align}
&\quad S(P_{W^rX^2Y^2})
\nonumber\\
&= \sum_{i\in[r]\setminus 2\mathbb{Z}}
[I(W_i;X_1|W^{i-1})+I(W_i;X_2|W^{i-1},X_1)]
\nonumber\\
&\quad+\sum_{i\in[r]\cap2\mathbb{Z}}[I(W_i;Y_1|W^{i-1},Y_2)
+I(W_i;Y_2|W^{i-1})]
\\
&= \sum_{i\in[r]\setminus 2\mathbb{Z}}
[I(W_i,Y_2;X_1|W^{i-1},Y_2)-\Delta_i
+I(W_i,X_1;X_2|W^{i-1},X_1)]
\nonumber\\
&\quad+\sum_{i\in[r]\cap2\mathbb{Z}}[I(W_i,Y_2;Y_1|W^{i-1},Y_2)
+I(W_i,X_1;Y_2|W^{i-1},X_1)-\Delta_i]
\\
 &=S(P_{U^rX_1Y_1})+S(P_{\bar{U}^rX_2Y_2})  -\sum_{i=1}^r\Delta_i
 \\
&= S(P_{U^rX_1Y_1})+S(P_{\bar{U}^rX_2Y_2})  -I(X_1;Y_2|W^r).
\end{align}
Then the claim \eqref{e91} follows.
%\begin{align}
%S
% &=I(W_2;Y_1|W_1Y_2)+I(W_2;Y_2|W_1)
% +I(W_1;X_1)+I(W_1X_1;X_2)
% \\
% &=I(W_2;Y_1|W_1Y_2)+I(W_2;Y_2|W_1X_%1)
% +I(W_1Y_2;X_1)+I(W_1X_1;X_2)
% \nonumber\\
% &\quad-I(Y_2;X_1|W_1W_2)
%  \\
%&=S(P_{U^2|XY})+S(P_{\bar{U}^2|XY}) % -I(X_1;Y_2|W_1W_2).
%\end{align}
%In the binary symmetric case, using Lemma~\ref{lem3}, we see that $\frac{R}{S-R}\le \frac{\rho^2}{1-\rho^2}$, as desired.

\end{proof}

\begin{remark}
Above, we followed the original method of proof proposed in a classical paper of Kaspi~\cite{kaspi1985two}, which essentially builds on Csisz\'ar-sum identity in multiuser information theory. This method has been used recently in testing for independence~\cite{xiang2013interactive} and common randomness extraction~\cite{liu2016common}. A similar method was applied in~\cite{BBCR10} to a problem of (approximately) reconstructing a function of two correlated iid strings.
\end{remark}

\subsection{Gaussian case}
To obtain the same lower bound in the Gaussian case, we can use the result for binary symmetric sequence and apply a central limit theorem argument. 

\begin{lemma}\label{lem11}
Let $X$ and $Y$ be jointly Gaussian with correlation $\rho$.
Consider any random variables $U^r$, $r\in \mathbb{Z}$, $|U^r|<\infty$,
satisfying 
\begin{align}
&U_i-(X,U^{i-1})-Y,\quad i\in [r]\setminus2\mathbb{Z},
\\
&U_i-(Y,U^{i-1})-X,\quad i\in [r]\cap2\mathbb{Z}.
\end{align}
% Let
% \begin{align}
% R &= \sum_{i\in [r]\cap2\mathbb{Z}}I(U_i;X|U^{i-1})
% +\sum_{i\in [r]\setminus
% 2\mathbb{Z}}I(U_i;Y|U^{i-1});
% \\
% S &= \sum_{i\in [r]\cap2\mathbb{Z}}I(U_i;Y|U^{i-1})
% +\sum_{i\in [r]\setminus
% 2\mathbb{Z}}I(U_i;X|U^{i-1}).
% \end{align}
Then
$R(P_{U^rXY}) \le \rho^2 S(P_{U^rXY})$.    
\end{lemma}
\begin{proof}
We claim the following continuity result:
If $P^{(t)}_{XY}$ converges to $P_{XY}$ in the total variation distance, then $R(P^{(t)}_{U^rXY})$ and $S(P^{(t)}_{U^rXY})$ converge to  $R(P_{U^rXY})$ and $S(P_{U^rXY})$ respectively, where $P^{(t)}_{U^rXY}:=P^{(t)}_{U^r|XY}P_{XY}$. 
The following will establish the lemma. Let $(A_l,B_l)_{l=1}^t$ be an iid\ sequence of binary symmetric random variables with correlation $\rho$, and put $X^{(t)}:=\frac{A_1+\dots+A_t}{\sqrt{t}}
+a_tN$
and $Y^{(t)}:=\frac{B_1+\dots+B_t}{\sqrt{t}}+a_tN'$,
where $N$ and $N'$ are standard Gaussian random variables, and $N$, $N'$, $(X^t,Y^t)$ are independent.
% By a version of the central limit theorem (see e.g.\ \cite{pro}, reproduced in \cite[Theorem~2.2]{bally})
%\cite[Theorem~4]{valiant2010}
%\uri{reviewr 2: "state the theorem you use at least, and how its assumptions are satisfied. As it stands, this proofs feels more like a sketch or outline than an actual proof." I see its commented now...}
By the central limit theorem, 
we can choose some $a_t=o(1)$ such that
the distribution of $(X^{(t)},Y^{(t)})$ converges to the Gaussian distribution $P_{XY}$ in total variation (Proposition~\ref{prop_clt}).
Now suppose that the claim is not true, then there exists $P_{U^r|XY}$ satisfying the required Markov chains and $|\mathcal{U}^r|<\infty$ such that $R(P_{U^rXY}) > \rho^2 S(P_{U^rXY})$.
The continuity claim implies that 
$R(P^{(t)}_{U^rXY}) > \rho^2 S(P^{(t)}_{U^rXY})$ for some $t$. 
However, using the data processing inequality of mutual information it is easy to see that $R(P_{U^rA^tB^t})>R(P^{(t)}_{U^rXY})$
and that 
$S(P_{U^rA^tB^t})<S(P^{(t)}_{U^rXY})$.
Thus $R(P_{U^rA^tB^t})>\rho^2 S(P_{U^rA^tB^t})$, which is in contradition with Lemma~\ref{lem3} and Lemma~\ref{lem4}.

It remains to prove the continuity claim. 
Note that for each $u^r$, $(x,y)\mapsto P_{U^r|XY}(u^r|x,y)$ is a measurable function taking values in $[0,1]$. 
Thus the convergence in total variation implies that $\lim_{t\to\infty} P_{U^r}^{(t)}(u^r)=P_{U^r}(u^r)$ and hence
\begin{align}
\lim_{t\to\infty}H_{P^{(t)}}(U^r)=H_P(U^r),
\label{e108}
\end{align}
where the subscripts of $H$ denote the distributions with respect to which the entropies are computed.
Moreover, 
$$
(x,y)\mapsto P_{U^r|XY}(u^r|x,y)\ln P_{U^r|XY}(u^r|x,y)$$ is also a bounded measurable function, 
so 
\begin{align}
&\quad \lim_{t\to\infty}\mathbb{E}\left[P_{U^r|XY}
(u^r|X^{(t)},Y^{(t)})\ln P_{U^r|XY}(u^r|X^{(t)},Y^{(t)})\right]
\nonumber\\
&=\mathbb{E}\left[P_{U^r|XY}
(u^r|X,Y)\ln P_{U^r|XY}(u^r|X,Y)\right],
\end{align}
and summing over $u^r$ shows that
\begin{align}
\lim_{t\to\infty}H_{P^{(t)}}(U^r|X,Y)=H_P(U^r|X,Y).
\label{e109}
\end{align}
Note that \eqref{e108} and $\eqref{e109}$ imply the convergence of $R(P_{U^rXY}^{(t)}):=I_{P^{(t)}}(U^r;X,Y)$.
Now,
\begin{align}
S(P_{U^rXY}^{(t)})=I_{P^{(t)}}(U^r;X)+I_{P^{(t)}}(U^r;Y)-I_{P^{(t)}}(U^r;X,Y),     
\end{align}
and hence it remains to show that \begin{align}
\lim_{t\to\infty}H_{P^{(t)}}(U^r|X)&=H_P(U^r|X),\label{e110}
\\
\lim_{t\to\infty}H_{P^{(t)}}(U^r|Y)&=H_P(U^r|Y).
\end{align}
By symmetry we only need to prove \eqref{e110}.
Let us construct a coupling of $P_{U^rXY}$ and $P_{U^rXY}^{(t)}$ as follows. First construct a coupling such that $(X^{(t)},Y^{(t)})=(X,Y)$ with probability $\delta_t:=\frac{1}{2}|P_{XY}^{(t)}-P_{XY}|$. 
Let $E$ be the indicator of the event $(X^{(t)},Y^{(t)})\neq (X,Y)$.
When $E=0$, generate $U^{r(t)}=U^r$ according to $P_{U^r|XY}(\cdot|X,Y)$.
When $E=1$, generate $U^{r(t)}$ according to $P_{U^r|XY}(\cdot|X^{(r)},Y^{(r)})$
and $U^r$ according to $P_{U^r|XY}(\cdot|X,Y)$ independently.
Then note that under either $P_{U^rXY}$ or $P^{(t)}_{U^rXY}$,
\begin{align}
|H(U^r|X)-H(U^rE|X)|&\le H(E).
\label{e112}
\end{align}
Moreover,
\begin{align}
& H(U^r,E|X)
\nonumber\\
&=H(U^r|X,E)
\\
&=\mathbb{P}[E=1]H(U^r|X,E=1)
+\mathbb{P}[E=0]H(U^r|X,E=0),
\end{align}
hence
\begin{align}
    |H(U^r,E|X)
-(1-\delta_t)H(U^r|X,E=0)|\le \delta_t\log|\mathcal{U}^r|.
\label{e115}
\end{align}
However, for any $\mathcal{A}\in \mathcal{X}$ and $u^r$, $$\frac{\mathbb{P}[U^r=u^r,X\in \mathcal{A},E=0]}{\mathbb{P}[X\in \mathcal{A},E=0]}=
\frac{\mathbb{P}[U^{r,(t)}=u^r,X^{(t)}\in \mathcal{A},E=0]}{\mathbb{P}[X^{(t)}\in \mathcal{A},E=0]},$$ 
implying that $P^{(t)}_{U^r|X=x,E=0}(u^r)=
P_{U^r|X=x,E=0}(u^r)$ for each $x$ and $u^r$, and hence $H_{P^{(t)}}(U^r|X,E=0)=H_P(U^r|X,E=0)$.
Thus \eqref{e112} and \eqref{e115} imply that 
\begin{align}
&\quad |H_{P^{(t)}}(U^r|X)-H_P(U^r|X)|
\nonumber\\
&\le 2\delta_t\log |\mathcal{U}^r|
+2\left[\delta_t\log\frac{1}{\delta_t}
+(1-\delta_t)\log\frac{1}{1-\delta_t}\right]
\end{align}
and \eqref{e110} follows since $\delta_t\to0$.
\end{proof}

\begin{remark}\label{rem_GSDPI}
In \cite{liu2016common}, we first computed the symmetric SDPI for the binary symmetric distribution and then proved the converse for secret key generation for the binary source.
Then we proved a converse for the Gaussian source using the following reduction argument:
Using a large block of binary symmetric random variables we can simulate a joint distribution converging to the Gaussian distribution in total variation. 
Thus the error probability of the operational problem cannot be too different under the simulated distribution and the true Gaussian distribution. 
In the present paper, we used a different argument to prove something stronger: the symmetric SPDI constant is equal to $\rho^2$ for the Gaussian distribution; this of course implies the converse for the operational problem.
\end{remark}

\section{Acknowledgments}
We thank anonymous reviewer for the idea of reducing the general case to $(\rho,0)$, cf. Corollary~\ref{crlry:anyTwoCorrs}. 
We note that this is precisely the method introduced by Rahman-Wagner \cite{rahman2012optimality}, though they only used it in a one-way setting.
We thank Yihong Wu, Mark Braverman, Rotem Oshman, Himanshu Tyagi and Sahasranand KR for useful discussions. This material is based upon work supported by the National Science Foundation CAREER award under grant agreement CCF-12-53205, the Center for Science of Information (CSoI),
an NSF Science and Technology Center, under grant agreement CCF-09-39370, the European Research Council, under grant agreement 639573, the Israeli Science Foundation, under grant agreement 1367/14, the MIT IDSS Wiener Fellowship and the Yitzhak and Chaya Weinstein Research Institute for Signal Processing.

\bibliography{bibtex_refs}
\bibliographystyle{alpha}

\end{document}